%% file: main.tex
\documentclass[11pt]{article}

\usepackage{graphicx}

\def\showauthornotes{0}

\def\showkeys{0}
\def\showdraftbox{0}
\def\showcolorlinks{1}
\def\usemicrotype{1}
\def\showfixme{0}

\usepackage[
backend=biber,
maxcitenames=50,
maxbibnames=99,
style=alphabetic,
sorting=nyt
]{biblatex}
\input{macros}

\title{Sparse High Dimensional Expanders via Local Lifts}
\author{
    Inbar Ben Yaacov\thanks{Department of Computer Science, Weizmann Institute of Science, Israel. email: inbar-ben.yaacov@weizmann.ac.il. 
    This project has received funding from the European Research Council (ERC) under the European Union’s Horizon 2020 research and innovation programme (grant agreement No. 819702) and from the Simons Foundation Collaboration on the Theory of Algorithmic Fairness.}
    \and
    Yotam Dikstein\thanks{Institute for Advanced Study, USA. email: yotam.dikstein@gmail.com. This material is based upon work supported by the National Science Foundation under Grant No. DMS-1926686.}
    \and
    Gal Maor\thanks{Tel Aviv University, Israel. email: galmaor@mail.tau.ac.il. Supported by Gil Cohen’s ERC grant
949499.}}
\date{\today}

\begin{document}

\maketitle

\begin{abstract}
    High dimensional expanders (HDXs) are a hypergraph generalization of expander graphs. They are extensively studied in the math and TCS communities due to their many applications. Like expander graphs, HDXs are especially interesting for applications when they are bounded degree, namely, if the number of edges adjacent to every vertex is bounded. However, only a handful of constructions are known to have this property, all of which rely on algebraic techniques. In particular, no random or combinatorial construction of bounded degree high dimensional expanders is known. As a result, our understanding of these objects is limited.

    The degree of an \(i\)-face in an HDX is the number of \((i+1)\)-faces that contain it.
    In this work we construct complexes whose higher dimensional faces have bounded degree.
    This is done by giving
    an elementary and deterministic algorithm that takes as input a regular \(k\)-dimensional HDX \(X\) and outputs another regular \(k\)-dimensional HDX \(\what{X}\) with twice as many vertices. While the degree of vertices in \(\what{X}\) grows, the degree of the \((k-1)\)-faces in \(\what{X}\) stays the same. As a result, we obtain a new `algebra-free' construction of HDXs whose \((k-1)\)-face degree is bounded.
    
    Our construction algorithm is based on a simple and natural generalization of the expander graph construction by Bilu and Linial \cite{BiluL2006}, which build
    expander graphs using lifts coming from edge signings.  
    Our construction is based on \emph{local lifts} of high dimensional expanders, where a local lift is a new complex whose top-level links are lifts of the links of the original complex. We demonstrate that a local lift of an HDX is also an HDX in many cases.

    In addition, combining local lifts with existing bounded degree constructions creates new families of bounded degree HDXs with significantly different links than before. For every large enough \(D\), we use this technique to construct families of bounded degree HDXs with links that have diameter \(\geq D\).
\end{abstract}

\input{sections/01introduction}
\input{sections/02preliminaries}
\input{sections/03construction}

\input{sections/04existence}

\input{sections/05randomizedalg}
\input{sections/06high-probability-construction}
\printbibliography
%\bibliographystyle{alpha}
%\bibliography{bib}
\end{document}

%% file: macros.tex
\usepackage{bbold}
\newcommand{\numvec}[1]{\mathbb{#1}}

\newcommand{\ovec}{\numvec{1}}
\newcommand{\indvec}[1]{\ovec_{#1}}
\newcommand{\mst}{\text{s.t. }}
\newcommand{\spec}[1]{\mathsf{Spec}\paren{#1}}

\newcommand{\what}[1]{\widehat{#1}}
\usepackage{mleftright}
\newcommand{\Iset}[2]{\mleft\{\, #1 \;\middle|\; #2 \,\mright\}}

\newcommand{\s}{\sigma}
\newcommand{\ldo}{\ldots}
\newcommand{\prodft}{\{\pm1\}}

\newcommand{\nt}{\notag}
\newcommand{\prin}[1]{\Pr\!\Brac{#1}}

\newcommand{\Ifrak}[1]{\mathfrak{#1}}
\newcommand{\dotdef}{\vcentcolon=}

\newcommand{\Isign}{\sign}

\newcommand{\Ibrace}[1]{\{#1\}}

\newcommand{\Icard}{c}
\providecommand{\RR}{\mathbb{R}}
\providecommand{\NN}{\mathbb{N}}

\usepackage{csquotes}
\usepackage[l2tabu, orthodox]{nag}

\usepackage{xspace,enumerate}

\usepackage[dvipsnames]{xcolor}

\usepackage[T1]{fontenc}
\usepackage[full]{textcomp}

\usepackage[american]{babel}

\usepackage{mathtools}

\usepackage{amsthm}
\usepackage{amsmath}

\ifnum\showcolorlinks=1
\usepackage[
colorlinks=true,
urlcolor=blue,
linkcolor=blue,
citecolor=OliveGreen,
]{hyperref}
\fi

\ifnum\showcolorlinks=0
\usepackage[
colorlinks=false,
pdfborder={0 0 0}
]{hyperref}
\fi

\usepackage[nameinlink, noabbrev, capitalize]{cleveref}
\newtheorem{theorem}{Theorem}[section]
\newtheorem*{theorem*}{Theorem}

\newtheorem*{proposition*}{Proposition}
\newtheorem{lemma}[theorem]{Lemma}
\newtheorem*{lemma*}{Lemma}
\newtheorem{corollary}[theorem]{Corollary}
\newtheorem*{corollary*}{Corollary}
\newtheorem*{conjecture*}{Conjecture}
\newtheorem{fact}[theorem]{Fact}
\newtheorem*{fact*}{Fact}

\newtheorem*{hypothesis*}{Hypothesis}

\theoremstyle{definition}
\newtheorem{definition}[theorem]{Definition}
\newtheorem*{definition*}{Definition}
\newtheorem{construction}[theorem]{Construction}

\theoremstyle{remark}
\newtheorem{claim}{Claim}[subsection] % was [theorem]
\newtheorem*{claim*}{Claim}
\newtheorem{remark}[theorem]{Remark}
\newtheorem*{remark*}{Remark}
\newtheorem{observation}[theorem]{Observation}
\newtheorem*{observation*}{Observation}

 \usepackage[letterpaper,
 top=1in,
 bottom=1in,
 left=1in,
 right=1in]{geometry}

\usepackage[varg]{pxfonts} % varg - uses nicer g,v,y,w

\usepackage[sc,osf]{mathpazo}
\usepackage{lmodern}

\ifnum\showkeys=1
\usepackage[color]{showkeys}
\fi

\usepackage{prettyref}

\newcommand{\savehyperref}[2]{\texorpdfstring{\hyperref[#1]{#2}}{#2}}

\newrefformat{eq}{\savehyperref{#1}{\textup{(\ref*{#1})}}}
\newrefformat{lem}{\savehyperref{#1}{Lemma~\ref*{#1}}}
\newrefformat{def}{\savehyperref{#1}{Definition~\ref*{#1}}}
\newrefformat{thm}{\savehyperref{#1}{Theorem~\ref*{#1}}}
\newrefformat{cor}{\savehyperref{#1}{Corollary~\ref*{#1}}}
\newrefformat{cha}{\savehyperref{#1}{Chapter~\ref*{#1}}}
\newrefformat{sec}{\savehyperref{#1}{Section~\ref*{#1}}}
\newrefformat{sbsc}{\savehyperref{#1}{Subsection~\ref*{#1}}}
\newrefformat{app}{\savehyperref{#1}{Appendix~\ref*{#1}}}
\newrefformat{tab}{\savehyperref{#1}{Table~\ref*{#1}}}
\newrefformat{fig}{\savehyperref{#1}{Figure~\ref*{#1}}}
\newrefformat{hyp}{\savehyperref{#1}{Hypothesis~\ref*{#1}}}
\newrefformat{alg}{\savehyperref{#1}{Algorithm~\ref*{#1}}}
\newrefformat{rem}{\savehyperref{#1}{Remark~\ref*{#1}}}
\newrefformat{item}{\savehyperref{#1}{Item~\ref*{#1}}}
\newrefformat{step}{\savehyperref{#1}{step~\ref*{#1}}}
\newrefformat{conj}{\savehyperref{#1}{Conjecture~\ref*{#1}}}
\newrefformat{fact}{\savehyperref{#1}{Fact~\ref*{#1}}}
\newrefformat{prop}{\savehyperref{#1}{Proposition~\ref*{#1}}}
\newrefformat{prob}{\savehyperref{#1}{Problem~\ref*{#1}}}
\newrefformat{claim}{\savehyperref{#1}{Claim~\ref*{#1}}}
\newrefformat{relax}{\savehyperref{#1}{Relaxation~\ref*{#1}}}
\newrefformat{red}{\savehyperref{#1}{Reduction~\ref*{#1}}}
\newrefformat{part}{\savehyperref{#1}{Part~\ref*{#1}}}
\newrefformat{ex}{\savehyperref{#1}{Example~\ref*{#1}}}
\newrefformat{para}{\savehyperref{#1}{Paragraph~\ref*{#1}}}
\newrefformat{ass}{\savehyperref{#1}{Assumption~\ref*{#1}}}
\newrefformat{question}{\savehyperref{#1}{Question~\ref*{#1}}}
\newrefformat{obs}{\savehyperref{#1}{Observation~\ref*{#1}}}

\newrefformat{const}{\savehyperref{#1}{Construction~\ref*{#1}}}%%%added by inbar

\newcommand{\Sref}[1]{\hyperref[#1]{\S\ref*{#1}}}

\usepackage{nicefrac}
\ifnum\usemicrotype=1
\usepackage{microtype}
\fi

\ifnum\showauthornotes=1
\newcommand{\Authornote}[2]{{\sffamily\small\color{red}{[#1: #2]}}}
\newcommand{\Authornotecolored}[3]{{\sffamily\small\color{#1}{[#2: #3]}}}
\newcommand{\Authorcomment}[2]{{\sffamily\small\color{gray}{[#1: #2]}}}
\newcommand{\Authorstartcomment}[1]{\sffamily\small\color{gray}[#1: }

\newcommand{\Authorfnote}[2]{\footnote{\color{red}{#1: #2}}}
\newcommand{\Authorfixme}[1]{\Authornote{#1}{\textbf{??}}}
\newcommand{\Authormarginmark}[1]{\marginpar{\textcolor{red}{\fbox{\Large #1:!}}}}
\else
\newcommand{\Authornote}[2]{}
\newcommand{\Authornotecolored}[3]{}
\newcommand{\Authorcomment}[2]{}
\newcommand{\Authorstartcomment}[1]{}

\newcommand{\Authorfnote}[2]{}
\newcommand{\Authorfixme}[1]{}
\newcommand{\Authormarginmark}[1]{}
\fi

\ifnum\showfixme=0

\fi

\usepackage{boxedminipage}

\newcommand{\paren}[1]{(#1)}
\newcommand{\Paren}[1]{\left(#1\right)}

\newcommand{\Brac}[1]{\left[#1\right]}

\newcommand{\abs}[1]{\lvert#1\rvert}
\newcommand{\Abs}[1]{\left\lvert#1\right\rvert}

\newcommand\sett[2]{\left\{ #1 \left| \; \vphantom{#1 #2} \right. #2  \right\}}
\newcommand{\set}[1]{\left\{#1\right\}}

\newcommand{\norm}[1]{\lVert#1\rVert}

\newcommand{\iprod}[1]{\langle#1\rangle}
\newcommand{\Iprod}[1]{\left\langle#1\right\rangle}

\def\dim{\mathrm{ dim}}

\newcommand{\Esymb}{\mathbb{E}}
\newcommand{\Psymb}{\mathbb{P}}

\DeclareMathOperator*{\E}{\Esymb}

\DeclareMathOperator*{\ProbOp}{\Psymb}

\renewcommand{\Pr}{\ProbOp}

\newcommand{\prob}[1]{\Pr \left[ {#1} \right] }
\newcommand{\Prob}[2][]{\Pr_{{#1}}\left[#2\right]} % use by \Prob[x]{event}
\newcommand{\ve}{\;\hbox{and}\;}

 \usepackage{dsfont}
\usepackage{mathrsfs}

\newcommand{\textparen}[1]{\text{(#1)}}

\ifx\because\undefined
\newcommand{\because}[1]{\textparen{because #1}}
\else
\renewcommand{\because}[1]{\textparen{because #1}}
\fi

\newcommand\bdot\bullet

\DeclareMathOperator{\Tr}{Tr}

\DeclareMathOperator{\poly}{poly}

\DeclareMathOperator{\sign}{sign}

\DeclareMathOperator*{\EE}{\mathbb{E}}

\renewcommand{\leq}{\leqslant}
\renewcommand{\le}{\leqslant}
\renewcommand{\geq}{\geqslant}
\renewcommand{\ge}{\geqslant}

\ifnum\showdraftbox=1

\else

\fi

\let\epsilon=\varepsilon

\numberwithin{equation}{section}

\newcommand{\MYstore}[2]{%
  \global\expandafter \def \csname MYMEMORY #1 \endcsname{#2}%
}

\newcommand{\MYload}[1]{%
  \csname MYMEMORY #1 \endcsname%
}

\newcommand{\MYnewlabel}[1]{%
  \newcommand\MYcurrentlabel{#1}%
  \MYoldlabel{#1}%
}

\newcommand{\MYdummylabel}[1]{}

\newcommand{\torestate}[1]{%
  \let\MYoldlabel\label%
  \let\label\MYnewlabel%
  #1%
  \MYstore{\MYcurrentlabel}{#1}%
  \let\label\MYoldlabel%
}

\newcommand{\restatetheorem}[1]{%
  \let\MYoldlabel\label
  \let\label\MYdummylabel
  \begin{theorem*}[Restatement of \prettyref{#1}]
    \MYload{#1}
  \end{theorem*}
  \let\label\MYoldlabel
}

\newcommand{\restatelemma}[1]{%
  \let\MYoldlabel\label
  \let\label\MYdummylabel
  \begin{lemma*}[Restatement of \prettyref{#1}]
    \MYload{#1}
  \end{lemma*}
  \let\label\MYoldlabel
}

\newcommand{\restateprop}[1]{%
  \let\MYoldlabel\label
  \let\label\MYdummylabel
  \begin{proposition*}[Restatement of \prettyref{#1}]
    \MYload{#1}
  \end{proposition*}
  \let\label\MYoldlabel
}

\newcommand{\restateclaim}[1]{%
  \let\MYoldlabel\label
  \let\label\MYdummylabel
  \begin{claim*}[Restatement of \prettyref{#1}]
    \MYload{#1}
  \end{claim*}
  \let\label\MYoldlabel
}

\newcommand{\restatecorollary}[1]{%
  \let\MYoldlabel\label
  \let\label\MYdummylabel
  \begin{corollary*}[Restatement of \prettyref{#1}]
    \MYload{#1}
  \end{corollary*}
  \let\label\MYoldlabel
}

\newcommand{\restatefact}[1]{%
  \let\MYoldlabel\label
  \let\label\MYdummylabel
  \begin{fact*}[Restatement of \prettyref{#1}]
    \MYload{#1}
  \end{fact*}
  \let\label\MYoldlabel
}

\newcommand{\restatedefinition}[1]{% % overwrite label command with dummy
\let\MYoldlabel\label 
\let\label\MYdummylabel 
\begin{definition*}[Restatement of \prettyref{#1}] 
    \MYload{#1} 
\end{definition*} 
\let\label\MYoldlabel 
} 

\newcommand{\restate}[1]{%
  \let\MYoldlabel\label
  \let\label\MYdummylabel
  \MYload{#1}
  \let\label\MYoldlabel
}

\let\origparagraph\paragraph
\renewcommand{\paragraph}[1]{\origparagraph{#1.}}

\allowdisplaybreaks

\let\pref=\prettyref

\newcommand{\bproof}[1]{\begin{proof}[Proof of \pref{#1}]}
\newcommand{\eproof}{\end{proof}}

\def\S{\mathcal{S}}

\def\d_1{d}

\newcommand{\remove}[1]{}

\usepackage[normalem]{ulem}

%% file: sections/01introduction.tex
\section{Introduction} \label{sec:intro}

Expander graphs are graphs that are well connected. These objects are studied extensively in computer science and mathematics \cite{HooryLW2006}, and since their discovery they have found numerous applications in complexity \cite{Dinur07,Reingold2008}, coding theory \cite{SipserS1996,TaShma2017,DinurELLM2022}, derandomization \cite{HooryLW2006,Goldreich2011} and more. Most of these applications rely on families of expander graphs that have a bounded degree. It is well known that random regular graphs are expanders, and many explicit bounded degree constructions are also in hand \cite{Margulis1973,LubotzkyPS1988,ReingoldVW2001,BiluL2006,MarcusSS2015}.

Recently, the study of \emph{high dimensional expanders} (HDXs) emerged (see surveys \cite{lubotzky2018high, gotlib2023no}). These are hypergraph analogues of expander graphs. 
While the full potential of high dimensional expanders is yet to be discovered, they are already important objects of study. High dimensional expanders, and especially bounded degree high dimensional expanders\footnote{ 
A family of HDXs is called \emph{bounded degree} if there is some \(M>0\) 
so that all vertices in every HDX in the family have degree at most \(M\). 
} 
have already yielded exciting applications in various areas such as locally testable codes \cite{DinurELLM2022,PanteleevK22,DinurLZ2023}, quantum complexity \cite{AnshuBN2023}, sampling and Markov chains \cite{DinurK2017,KaufmanO-RW20}, agreement testing \cite{DinurK2017,DiksteinDL2024,BafnaLM2024}, high dimensional geometry and topology \cite{Gromov2010,FoxGLNP2012}, pseudorandomness \cite{CRT23} and random (hyper)graph theory \cite{LinialM2006,MeshulamW09}.

The specific family of high dimensional expanders used in many of the aforementioned applications is tailor-made to satisfy other desired properties, in addition to high dimensional expansion. For example, the local neighborhoods in the high dimensional expanders used in \cite{DinurELLM2022,PanteleevK22,DinurLZ2023} are tailored so that one can define a small locally testable code on them; the high dimensional expanders in \cite{Gromov2010,DiksteinDL2024,BafnaLM2024} also have a vanishing \(1\)-cohomology over certain group coefficients.

However, constructing bounded degree high dimensional expanders (for arbitrarily small spectral expansion of the links) is still a serious challenge. No random model for bounded degree high dimensional expanders is known, and all deterministic constructions known use non-trivial algebraic techniques. The fact that we have only a handful of bounded degree constructions to choose from, makes these objects difficult to understand and to work with.
We believe that many further applications await us once we learn how to diversify these constructions, in the same way that many of the above-mentioned applications of expander graphs grew out of more varied expander constructions that were discovered.

Nowadays, 
all known constructions of HDXs rely on algebraic techniques, including quotients of the Bruhat Tits buildings \cite{Ballantine2000,CartwrightSZ2003,Li2004,LubotzkySV2005b,Dikstein2022,DiksteinDL2024,BafnaLM2024} and coset complexes \cite{KaufmanO18,FriedgutI2020,ODonnellP2022} (see also \cite{HarshaS2019} for a more elementary analysis of some of these HDXs). There have been attempts at constructing bounded degree HDXs with combinatorial tools, but all these constructions fall short either in bounded degreeness \cite{Golowich2023,LiuMSY2023} or in their local spectral expansion 
\cite{Conlon2019,ConlonTZ2020,ChapmanLP2020,LiuMY2020,Golowich2021}.

In particular, it is an important open question whether an algorithm à la Zig-Zag product \cite{ReingoldVW2001} exists for bounded degree high dimensional expanders. That is, an algorithm that given a bounded degree high dimensional expander as input, outputs another high dimensional expander with more vertices and the same bound on the degree and spectral expansion.

As an intermediate result, in this work we develop an algorithm that takes a high dimensional expander as input, and outputs another high dimensional expander with more vertices, the same bound on spectral expansion and the same bound on \emph{the degree of high dimensional faces} (but not on the degree vertices). This algorithm is entirely combinatorial, relying only on the theory of graph covers initiated by \cite{AmitL2002,BiluL2006}. While families of complexes constructed via such an algorithm are not sufficient for applications that require the vertex degree to be bounded, we view this as a stepping stone towards an `algebra free' construction of bounded degree HDXs. One exception to this is the recent work by \cite{AP23}, which analyzes a variant of the well known Glauber dynamics (or up-down) random walk on HDXs. The walk that \cite{AP23} analyzes has bounded degree if and only if the underlying HDX is \((k-1)\)-bounded degree.

\subsection{Preliminaries on High Dimensional Expanders}

To better understand our results, 
let us introduce some standard definitions and notation on simplicial complexes (see \pref{sbsc:hdx-prelims} 
for a more elaborated definitions). 
A simplicial complex is a hypergraph that is downwards closed to containment. A simplicial complex is \(k\)-dimensional if the largest hyperedge in the complex is of size \((k+1)\). We denote by \(X(\ell)\) the sets (aka faces) of size \(\ell+1\). Let \(X\) be a \(k\)-dimensional simplicial complex.  

The degree of a face \(\sigma \in X(\ell)\), is \(d(\sigma) \dotdef \abs{\sett{\tau \in X(\ell+1)}{\tau \supseteq \s}}\).
A family of complexes \(\set{X_i}_{i=0}^\infty\) is \emph{\(\ell\)-bounded degree} if there exists an \(M > 0\) that bounds the degrees of \emph{all \(\ell\)-faces} across \emph{all the complexes simultaneously}. That is, for every \(i\) and any \(\s \in X_i(\ell)\), \(d(\s) \leq M\). We say that a family of complexes are bounded degree if they are \(0\)-bounded degree.
We say that a complex is \((d_0,d_1,\dots,d_{k-1})\)-regular if for every \(\s \in X(\ell)\), \(d(\s)=d_{\ell}\). 

In this paper we are mainly interested in the local spectral expansion definition of high dimensional expanders 
(see \cite{lubotzky2018high} for a survey on other definitions). For this we need to define `links', the generalization of vertex neighborhoods in graphs. For a face \(\sigma \in X\), the link 
of $\s$ is the simplicial complex \(X_\sigma = \sett{\tau \setminus \sigma}{\sigma \subseteq \tau \in X}\). We will be interested in the graph structure underlying the complex and its links. The \(1\)-skeleton of \(X\) is the graph whose vertices are \(X(0)\) and whose edges are \(X(1)\).

\begin{definition}[High dimensional expander]
    For \(\lambda > 0\) we say that \(X\) is a \(\lambda\)-two sided (one sided) high dimensional expander if for every \(\ell \leq k-2\) and \(\sigma \in X(\ell)\), the \(1\)-skeleton of \(X_\sigma\) is a \(\lambda\)-two sided (one sided) spectral expander.
\end{definition} 

\subsection{Our Results}
Our results are based on the notion of a graph lifts. We say a graph \(\what{G} = (\what{V},\what{E})\) is a lift of a graph \(G = (V,E)\) if there exists a graph homomorphism \(\phi:\what{V} \to V\) such that for every \(\hat{v} \in \what{V}\), the mapping \(\phi\) is a bijection on the neighborhood of \(\hat{v}\). Intuitively, a lift of a graph \(G\) is a large graph \(\what{G}\) that locally looks the same as \(G\). Graph lifts are essential in many constructions of expander graphs \cite{AmitL2002,BiluL2006,MarcusSS2015}, and our construction builds on the beautiful work of \cite{BiluL2006}. We elaborate more on this below.

Our main result is a construction algorithm that maintains both expansion and the degree of the \((k-1)\)-faces of a regular complex. 
This algorithm takes as input a \((d_0,d_1,\dots,d_{k-1})\)-regular \(\lambda\)-high dimensional expander \(X\), and outputs a \((2d_0,2d_1,\dots,2d_{k-2},d_{k-1})\)-regular \(\what{X}\) with twice as many vertices that is also a \(\lambda\)-HDX (even though the number of intermediate faces grows like \(|\what{X}(i)| = 2^i |X(i)|\) for \(i \leq k-1\)). This algorithm uses the notion of random lifts \cite{BiluL2006}, and in particular, it requires no algebraic machinery for the construction nor the analysis. More formally, this is the theorem we prove.

\begin{theorem}[See \pref{thm:rand-alg-lll-full} for a more precise statement]\label{thm:main-one-step}
    There exists a randomized algorithm \(\mathcal{A}\) that takes as input a \(k\)-dimensional complex \(X_0\) and an integer \(i \geq 1\), runs in expected time \(\poly((2^i|X_0(0)|)^k)\) at most, and outputs a \(k\)-dimensional complex \(X_i\) with \(2^i |X_0(0)|\) vertices. The algorithm has the following guarantees.
    \begin{enumerate}
        \item If \(X_0\) is a \((d_0,d_1,\ldo,d_{k-1})\)-regular \(\lambda\)-two sided high dimensional expander,
        then \(X_i\) is a \((2^i d_0,\dots,2^i d_{k-2},d_{k-1})\)-regular \(\lambda'\)-two sided high dimensional expander where 
        \[\lambda' = O \left ( \max  \set{\lambda\Paren{1+\log \frac{1}{\lambda}}, \sqrt{\frac{\log^3 d_{k-1}}{d_{k-1}}}}\right ).\]      
        \item For every \(\hat{\s} \in X_i(k-2)\), the link \((X_i)_{\hat{\s}}\) is a lift of $(X_{i-1})_\s$ for some $\s\in X_{i-1}(k-2)$.
    \end{enumerate}
    For every \(j \leq k-2\), \(|X_i(j)|=2^{(j+1)i}|X_0(j)|\).
\end{theorem}

There are various complexes in hand that one can use as the input to this algorithm. These include the complete complex, the complexes from \cite{LubotzkyLR2019}, and even complexes from bounded degree families that are regular, such as those constructed by \cite{FriedgutI2020}. 

We give two proofs to \pref{thm:main-one-step}, building on the techniques developed by \cite{BiluL2006} to analyze lifts in expander graphs, and extend them to high dimensional expanders.

While most of the work in \cite{BiluL2006} regards random lifts of graphs, they also show how to deterministically find expander graphs using lifts. Building on their method, we also give a deterministic algorithm for finding the complexes in \pref{thm:main-one-step}, albeit under some more assumptions on the input \(X_0\). This provides a \emph{deterministic, polynomial time and elementary} construction of a family of \((k-1)\)-bounded \(k\)-dimensional high dimensional expanders.

\begin{theorem}[See \cref{cor:explicit_hdx}] \label{thm:main-construction-deterministic}
There exists a deterministic algorithm \(\mathcal{B}\) that takes as input a \(k\)-dimensional complex \(X_0\) and an integer \(i \geq 1\), runs in time \(\poly((2^i|X_0(0)|)^k)\) at most, and outputs a \(k\)-dimensional complex \(X_i\) with \(2^i |X_0(0)|\) vertices. The algorithm has the following guarantees.
    \begin{enumerate}
        \item If \(X_0\) is a \((d_0,d_1\ldo,d_{k-1})\)-regular \(\lambda\)-two sided high dimensional expander, such that \(d_{k-1}>2^{10k}\) and  $|X_0(k-2)|\le \left(d_{k-2}\right)^{10k}$,
        then \(X_i\) is a \((2^i d_0,\dots,2^i d_{k-2},d_{k-1})\)-regular \(\lambda'\)-two sided high dimensional expander where 
        \[\lambda' = O \left ( 2^{5k} \max  \set{\lambda\Paren{1+\log \frac{1}{\lambda}}, \sqrt{\frac{\log^3 d_{k-1}}{d_{k-1}}}}\right ).\]
        \item For every \(\hat{\s} \in X_i(k-2)\), the link \((X_i)_{\hat{\s}}\) is a lift of $(X_{i-1})_\s$ for some $\s\in X_{i-1}(k-2)$.
\end{enumerate}
For every \(j \leq k-2\), \(|X_i(j)|=2^{(j+1)i}|X_0(j)|\).
\end{theorem}

Not only is our construction deterministic, but it is also simple and versatile; one can apply it to various kinds of high dimensional expanders, and the family of HDXs obtained by doing so is changes according to the initiating HDX given at the beginning of the process.

\subsection{Comparing to Random Constructions of HDXs}
In the graph case the configuration model yields regular and bounded degree expanders. In contrast, there is no immediate generalization of this model to higher dimensions, that leads to bounded degree HDXs, even if one only wishes to bound the degrees of higher dimensional faces. If one settles for logarithmic degree, then one could use the \cite{LinialM2006} random model to construct random HDXs.
The degree of the top-level faces of these complexes is \(O(\log n)\), where \(n\) is the number of vertices, and the degree of the lower dimensional faces is polynomial in \(n\). For \(2\)-dimensional complexes, the random geometric model in \cite{LiuMSY2023} offers an improvement to the vertex degree that one gets from \cite{LinialM2006}, but it is still polynomial.

It is tempting to try and adapt the \cite{LinialM2006} model for constructing $(k-1)$-bounded degree HDXs, but doing so in a straightforward manner falls short of achieving that. The work by \cite{LubotzkyLR2019} found an appropriate generalization of the random model that gives \((k-1)\)-bounded degree HDXs, utilizing the breakthrough work of \cite{Keevash2014} on Steiner systems. In their model, one takes a complete \((k-1)\)-skeleton and samples \(k\)-faces by sampling random Steiner systems on this complex.

Our construction sidesteps this difficulty by taking a different approach; it uses random \emph{local lifts} of HDXs (presented in the following subsection) instead of trying to construct random ones from scratch. 
In this setting, the high dimensional case behaves more similar to the \(1\)-dimensional one - our work shows that random local lifts of HDXs are HDXs with high probability. Of course, this requires an appropriate modification of the lift notion to local lifts.

\subsection{The Construction}
We now dive into the heart of \pref{thm:main-one-step}. Our construction builds upon the work of random lifts of graphs studied in \cite{BiluL2006}. Random lifts of expanders have been subject to extensive research (e.g., \cite{AmitLMR2001,BiluL2006,AgarwalKM13,MarcusSS2015,Bor20}). However, in this work, we do not try to lift the complex itself. Instead, we construct another complex where the \emph{$(k-2)$-links} are lifts of links in the original complex. We call such a complex a \emph{local lift}.

Let us first explain the idea behind the work of \cite{BiluL2006}. 
Their work suggests a construction of bounded degree family of
expander graphs 
$\set{G_i}_{i=0}^\infty$,
where for every $i$, \(G_{i+1}\) is a lift of \(G_i\). 
The fact that the maximal degree of \(G_{i+1}\) is equal to the maximal degree of \(G_i\) promises that the sequence is bounded degree. Therefore, one only needs to worry about expansion.

The work \cite{BiluL2006} studies random lifts sampled using signings on the edges of a graph $G=(V,E)$, that is, functions \(f:E \to \set{\pm 1}\). 
Given such a signing \(f\), one can construct the following lift \(\what{G} = (\what{V}, \what{E})\) 
by setting
\(\what{V} = V \times \set{\pm 1}\) and \(\set{(v,i),(u,j)} \in \what{E}\) if \(\set{v,u} \in E\) and \(i\cdot j = f(\set{u,v})\).

The work of \cite{BiluL2006} analyzes when a lift \(\what{G}\) obtained by random signing \(f\) is an expander graph. 
They give a proof (based on the Lov\'asz Local Lemma) that every expander \(G\) has
such a ``good'' signing. 
They also provide a deterministic algorithm to construct such a lift using the conditional probabilities method \cite{alon2016probabilistic}.

Our construction generalizes this idea, only instead of lifts coming from edge 
signings, we define \emph{local lifts} coming from \emph{face}-signings.

Let \(X\) be a \(k\)-dimensional simplicial complex and let \(f:X(k) \to \set{\pm 1}\). Define the \(k\)-dimensional complex \(\what{X}\) (where \(f\) is implicit in the notation) as a complex whose vertices are \(\what{X}(0) = X(0) \times \set{\pm 1}\), and whose \(k\)-faces are
\[
    \what{X}(k) = \sett{\left\{v_0^{j_0},v_1^{j_1},\dots,v_k^{j_k}\right\}}{ \set{v_0,v_1,\dots,v_k} \in X(k) \ve \prod_{i=0}^k j_i = f(\set{v_0,v_1,\dots,v_k})}.
\]
For \(1\leq \ell \leq k-1\) the \(\ell\)-faces are independent of the second coordinate, that is,
    \[ \what{X}(\ell) = \sett{\left\{v_0^{j_0},v_1^{j_1},\dots,v_\ell^{j_\ell}\right\}}{\set{v_0,v_1,\dots,v_\ell} \in X(\ell)}.\]

Obviously, the underlying graph of \(\what{X}\) is \emph{not} a lift of the underlying graph of \(X\). Indeed, the degree of each vertex is doubled. However, for every \(\hat{\sigma} \in \what{X}(k-2)\), we show that \(\what{X}_{\hat{\sigma}}\) is isomorphic to a lift of \(X_\sigma\) (where \(\sigma = \sett{v}{v^j \in \hat{\sigma}}\)).

Indeed, let us assume for simplicity that \(k=2\). Consider the link of a vertex \(v^j \in \what{X}(0)\) and define the function \(g:X_v(1) \to \set{\pm 1}\) by \(g(uw)= j \cdot f(uw)\). We claim that \(\what{X}_{v^j}\) is the cover \(g\) induces on \(X_v\). It is easy to see that its vertices are \(\what{X}_{v^j}(0) = X_v(0) \times \set{\pm 1}\), since the vertices in \(\what{X}_{v^j}\) correspond to edges 
in
\(\what{X}(1)\). These are precisely all \(u^1,u^{-1}\) where \(u \in X_v(0)\).

The edges are more delicate. Edges \(\set{u^{j'},w^{j''}} \in \what{X}_{v^j}(1)\) correspond to triangles \(\set{v^j,u^{j'},w^{j''}} \in \what{X}\). Indeed, such a triangle is in \(\what{X}\) if and only if:
\begin{enumerate}
    \item \(\set{u,w} \in X_v(1)\) and,
    \item \(j \cdot j' \cdot j'' = f(\set{v,u,w})\).
\end{enumerate}
The second item occurs if and only if \(j \cdot j' \cdot j'' = g(uw)\). Hence,  \(\what{X}_{v^j}\) is the cover \(g\) induces on \(X_v\).

As mentioned above, we give two proofs that signings \(f\) so that \(\what{X}\) is a high dimensional expander exist. 
The first proof is based on the Lov\'asz Local Lemma and follows the argument in \cite[Lemma 3.3]{BiluL2006}, and generalizes it so that multiple links may be taken into account simultaneously. The second proof is based on a different way to use the Lov\'asz Local Lemma (together with other results from \cite{BiluL2006}) which we find simpler, to deduce high dimensional expansion. This proof, while more restrictive on the link sizes, can be combined with the algorithmic version of the Lov\'asz Local Lemma \cite{MoserT2010}, to prove \pref{thm:main-one-step}. Afterwards, we show that if the links of the complex \(X\) are already dense, then the derandomization technique in \cite{BiluL2006} works for high dimensional expanders, and we can obtain a deterministic construction for \((k-1)\)-bounded \(k\)-dimensional high dimensional expanders, proving \pref{thm:main-construction-deterministic}.

\subsection{Understanding Vertex vs. Edge Degree in Bounded-Degree Constructions}
We can use \pref{thm:main-one-step} to diversifying links in other existing bounded degree constructions, and thus gain more understanding on how possible high dimensional expanders may look like. 
For simplicity, let us stick to the \(2\)-dimensional case, and consider the question \emph{how small could \(d_1\) be given \(d_0\) in a \((d_0,d_1)\)-regular high dimensional expander?} 

Let us consider the behavior of \(d_0\) and \(d_1\) in the known bounded degree constructions \cite{Ballantine2000,CartwrightSZ2003,Li2004,LubotzkySV2005b,KaufmanO18,FriedgutI2020,Dikstein2022,DiksteinDL2024,BafnaLM2024}. 
In all 
the above, \(d_0\) grows to infinity as \(\lambda\) goes to \(0\), and \(d_1 = \poly(d_0)\)\footnote{Technically most of the constructions above are not \emph{regular}, only bounded degree, so \(d_0\) and \(d_1\) should be average values, but we ignore this point for the sake of presentation.}. In other words, the links themselves are ``locally'' dense. 
A natural question to ask is whether the lower bound of $d_1 \geq d_0^{\Omega(1)}$ is necessary for \emph{bounded degree} constructions. In expander \emph{graphs} it is well known that one can increase the size of the graph without increasing the bound on the degree, but this is not the behavior in the known bounded-degree HDX constructions.  

We note that if one allows \(d_0\) to tend to infinity with \(n\), rather than staying constant, then works such as \cite{LubotzkyLR2019} (and also infinite families of complexes constructed by iteratively applying \pref{thm:main-construction-deterministic}) show that this is false. But this question is more interesting when its bounded degree.

\pref{thm:main-construction-deterministic} gives a negative answer to this question in the \(2\)-dimensional case, by proving the following.

\begin{theorem}\torestate{ \label{thm:diameter-intro}
    For every \(\lambda > 0\) and any sufficiently large \(M > 0\), there exists an infinite family of \(2\)-dimensional \(\lambda\)-two sided high dimensional expanders that are \((d_0,\exp(\poly(\frac{1}{\lambda})))\)-regular, for \(M \leq d_0 \leq 2M\).
    
    In particular, for every large enough \(D>0\), there exists an infinite family of \(2\)-dimensional \(\lambda\)-two sided HDXs such that the diameter in every link \(X_v\), is at least \(D\).}
\end{theorem}

We stress that \(d_1=\exp(\poly(\frac{1}{\tilde{\lambda}}))\) depends only on the spectral expansion and not on the number of vertices or \(d_0\).

\subsection{Related Work}
\paragraph{Bounded degree HDX} As discussed above, all known constructions of bounded degree high dimensional expanders use algebraic techniques. The first bounded degree high dimensional expanders for arbitrarily small \(\lambda > 0\) was by \cite{Ballantine2000}. This was followed by many other works that aimed to construct the high dimensional equivalent to Ramanujan graphs \cite{CartwrightSZ2003,Li2004,LubotzkySV2005a,LubotzkySV2005b}. All these constructions are quotients of \(\tilde{A}\)-type Bruhat Tits buildings. The work by \cite{Dikstein2022} used this building together with complex lifts to construct other high dimensional expanders. Recently, high dimensional expanders that come from \(\tilde{C}\) Bruhat Tits buildings were also constructed and studied \cite{ChapmanL2023,DiksteinDL2024,BafnaLM2024}. A second type of constructions come from coset complexes, first studied by \cite{KaufmanO18}. More complexes of this type were constructed by \cite{FriedgutI2020,ODonnellP2022}. We mention that the work by \cite{HarshaS2019} simplified the analysis of these coset complexes, and gave a description of the complexes in \cite{KaufmanO18} in relatively elementary means (albeit still relying on algebraic methods).

Interestingly, \cite{LubotzkyLR2019} give a randomized construction of a \((k-1)\)-bounded degree family of \(\lambda\)-HDXs for arbitrarily small \(\lambda > 0\). This construction is based on random Steiner systems and given in the breakthrough result of \cite{Keevash2014}. The underlying \((k-1)\)-skeletons of the complexes in that family are complete.

There are other bounded degree constructions \cite{Conlon2019,ConlonTZ2020,ChapmanLP2020,LiuMY2020,Golowich2021}. These constructions have various mixing properties, but none of them are \(\lambda\)-HDXs for \(\lambda < \frac{1}{2}\) (where \(\lambda\) is normalized between \(0\) and \(1\)). There are other constructions of \(\lambda\)-HDXs for \(\lambda < \frac{1}{2}\), which are not bounded degree, but are still non-trivially sparse. These include \cite{Golowich2023} - based on Grassmann posets, and \cite{LiuMSY2023} - based on random geometric graphs of the sphere. 

Finally, we comment that previous works also considered the possible degrees \((d_0,d_1)\) possible in a high dimensional expanders. The work by \cite{FriedgutI2020} used irregular algebraic constructions of bounded degree \(\lambda\)-HDXs and `regularized' them, thus showing that there exists bounded degree HDXs that are regular for arbitrarily small \(\lambda\). The work by \cite{ChapmanLP2020} gives a lower bound on the expansion of the underlying graph of the complex in terms of \((d_0,d_1)\), but this lower bound does not rule out (or construct) such HDXs with \(d_1 \ll d_0\). 

\paragraph{Graph and HDX lifts}The study of random graph lifting was initiated in \cite{AmitL2002}. Random lifts from signings of expanders were studied in \cite{BiluL2006} where it was proven that with high probability they are also expanders. This was extended to larger lifts as well \cite{Oliveira2009,AddarioG2010}. Friedman showed that random lifts of Ramanujan graphs are nearly Ramanujan \cite{Friedman2008} (see also \cite{Bor20}). In the seminal paper by \cite{MarcusSS2015}, Ramanujan bipartite graphs were constructed by using graph lifts. One can also define lifts of simplicial complexes. Most known bounded degree high dimensional expanders are constructed using a dual notion of lifts - that is, taking quotients of an infinite object \cite{Ballantine2000,CartwrightSZ2003,Li2004,LubotzkySV2005a,LubotzkySV2005b,KaufmanO18}, in a way such that the infinite object is a lift of the complex that is constructed. \cite{Dikstein2022} studied taking random lifts of simplicial complexes as in \cite{BiluL2006}, but the construction there needed the use of algebraic techniques as well.

\subsection{Open Questions}
\paragraph{Combinatorial constructions}As we mention 
earlier, there is no construction of bounded degree high dimensional expanders that does not rely on non-trivial algebraic techniques. As an intermediate step towards such a construction, can one give a construction of \(k\)-dimensional simplicial complexes that are \((k-2)\)-bounded degree (or \(i\)-bounded degree for any \(i < k-1\))?

\paragraph{Links with other properties}Fix a vertex set \([n]\) and graphs \(\set{G_i}_{i=1}^n\) (one graph for every vertex \(i \in [n]\)). It is interesting to understand whether there exists a graph whose vertex set is \(V=[n]\), and such that the neighborhood of every vertex \(i \in V\) is (isomorphic to) \(G_i\). The structure of such graphs is an extensive topic of study, especially in the case where all \(G_i\)'s are equal (see, e.g., \ \cite{BrownC1975,BlokhuisB1989,MarquezdMNR2003}). One of the major components in the works \cite{DinurELLM2022,PanteleevK22} that constructed asymptotically good locally testable codes and quantum codes, is a construction of graphs that locally look like a neighborhood of a graph product, but globally have improved expansion properties.

In this work, we propose a technique that addresses a related problem. Given a graph \(G\) (which is the one skeleton of a regular \(2\)-dimensional complex \(X\)), we find a graph \(\what{G}\) where every neighborhood in \(\what{G}\) is a random (or deterministic) \(2\)-lift of a corresponding neighborhood in \(G\). Is there a technique that allows us to do so for \emph{any} set of \(2\)-lifts of the respective vertex neighborhoods?

\paragraph{Other notions of expansion}In this paper we mainly deal with local spectral expansion, but other definitions of high dimensional expansion also exist. Most notable is the notion of coboundary expansion defined independently in \cite{LinialM2006} and \cite{Gromov2010}. This notion is important for many applications of high dimensional expanders such as code construction \cite{DinurELLM2022}, topological expansion \cite{Gromov2010} and property testing \cite{KaufmanL2014,GotlibK2022,DiksteinD2023agr,BafnaM2023}. Does a local lift maintain coboundary expansion? If not, is it maintained in interesting special cases?

\paragraph{Better local spectral expansion} Works following \cite{BiluL2006} such as \cite{MarcusSS2015,Bor20} improved the bounds on the spectrum of lifts of regular graphs. Can one construct local lifts of regular high dimensional expanders that are also \emph{Ramanujan}?

\subsubsection{Organization of This Paper}
The necessary preliminaries are given in \pref{sec:prelims}. We describe local lifts in \pref{sec:construction} and describe some of their basic properties. In \pref{sec:existence} we show existence of good local lifts by modifying a Lov\'asz Local Lemma argument by \cite{BiluL2006}. In \pref{sec:rand-alg} we prove \pref{thm:main-one-step} using the algorithmic Lov\'asz Local Lemma \cite{MoserT2010} and derive \pref{thm:diameter-intro}. In \pref{sec:deterministic} we show that the method of derandomization in \cite{BiluL2006} could be generalized to our case as well and prove \pref{thm:main-construction-deterministic}. 

%% file: sections/02preliminaries.tex
\section{Preliminaries} \label{sec:prelims}
Unless explicitly stated, all logarithms are with base \(2\). The \(\ln\) function is a logarithm with base \(e\). We write \(A \sqcup B\) to denote a disjoint union of sets $A,B$. The For \(n \geq 0\) we write \([n] = \set{0,1,\dots,n}\). For a square matrix (or equivalently, a linear operator on a finite vector space), we write \(\norm{A}\) to denote the operator norm.

\subsection{Graphs}
Let $G=(V,E)$ be a graph. For $u,v\in V$ we write $\Gamma(v)$ for the set of $v$'s neighbors in $G$ and $u\sim v$ if $u$ and $v$ are neighbors. The \emph{indicator vector} of a set $S\subseteq V$, denoted by $\indvec{S}$, is \(\indvec{S}:V \to \set{0,1}\) with $\indvec{S}(v)=1 \iff v\in V$. For two sets $S,T\subseteq V$ we write $E_G(S,T)$ for the set of edges in $G$ between $S$ and $T$. The \emph{graph induced by $S$ and $T$} is  $G'=(S\cup T,E_G(S,T))$. 

For a \(d\)-regular graph we denote \(\ell_2(V) = \set{f:V \to \RR}\) endowed with the inner product
\[\iprod{f,g} = \sum_{v \in V}f(v)g(v).\]

\subsubsection{Expander Graphs}\label{sbsc:prelim-expanders}
Expander graphs are graphs with good connectivity properties. There are many equivalent ways to define expanders \cite{HooryLW2006}. In this manuscript we focus on \emph{spectral expansion}.

Let $G=(V,E)$ be a $d$-regular graph. The \emph{random walk matrix} of $G$ is a matrix $A \in \mathbb{R}^{V \times V}$ defined by $A(u,v)=\frac{1}{d}$ if $u\in\Gamma(v)$ and otherwise $0$. Equivalently, it corresponds to the random walk operator \(A : \ell_2(V) \to \ell_2(V)\) with \(Af(v) = \frac{1}{d}\sum_{u \in \Gamma(v)}f(u)\). We abuse notation and use \(A\) for both the matrix and the random walk operator it represents. We sometimes denote this operator by \(A_G\) when \(G\) is unclear from the context.

The operator \(A\) is self adjoint with respect to the inner product. Therefore, it has an orthonormal basis of real-valued eigenvectors, where the eigenvalues are denoted by 
    $1=\lambda_1 \ge \lambda_2 \ge \cdots \ge \lambda_n$.
We elaborate and write $\lambda_i(G)$ when the graph in question is unclear from the context.
The \emph{spectrum of $G$} is the spectrum of its random walk matrix and is denoted by $\spec{G}$.

\begin{definition}[spectral expander]
    For $\lambda\in[0,1]$ we say that $G$ is \emph{$\lambda$-two sided (resp.\ one sided) spectral expander} (or \emph{expander} for short) if $\lambda\ge \max\{\lambda_2,|\lambda_n|\}$ (resp.\ $\lambda\ge \lambda_2$).
\end{definition}

\subsubsection{Tensor Product}
Let $G, H$ be any graphs. The \emph{tensor product} of $G$ and $H$, denoted by $G \otimes H$, is the graph with vertices $V(G)\times V(H)$, and edges $E(G \otimes H)=\sett{(a,b)(a',b')}{ \set{a,a'}\in E(G) \ve \set{b,b'}\in E(H)}$.

The following fact is well known.
\begin{fact}\label{fact:eigenvalues-of-tensor-product}
    Let \(G,H\) be graphs. If \(H,G\) are \(\lambda,\lambda'\)-two sided spectral expanders respectively, then \(G \otimes H\) is a \(\max\{\lambda,\lambda'\}\)-two sided spectral expander. Moreover, if \(H\) is a \(\lambda\)-two sided spectral expander and \(G\) is a \emph{\(\lambda'\)-one sided} spectral expander, then \(G \otimes H\) is a \(\max\{\lambda,\lambda'\}\)-one sided spectral expander.  
\end{fact}

\subsection{Graph Lifts}
Graph lifts are an important notion, studied first by \cite{AmitL2002,AmitLMR2001} (although the notion of lifts themselves is a classical notion in algebraic topology known for about a century).
\begin{definition}[lift]
    For finite, connected and simple graphs $G$ and $\what{G}$, a \emph{lift} (also known as a covering map) $\phi: \what{G} \to G$ is a graph homomorphism with the property that for all $\hat{v}\in V(\what{G})$, \(\phi\) maps the neighborhood of \(\hat{v}\) in \(\what{G}\) onto the neighborhood of \(\phi(\hat{v})\) in \(G\). Finally, we say that $\what{G}$ is an \emph{$\ell$-lift of $G$} if there exists an \(\ell\)-to-\(1\) covering map $\phi:\what{G} \to G$.
\end{definition}

One way to construct a \(2\)-lift is to use a signing function on the edges as follows.
\begin{definition}[Function induced lift]\label{def:function-induced-lift}
    Let \(G=(V,E)\) be a graph and let \(f:E \to \set{\pm 1}\) be a signing. The \emph{\(f\)-induced lift} \(\what{G} = \what{G}^f\) is the graph whose vertices are
    \[\what{V} = V \times \set{ \pm 1} = \sett{v^j}{v \in V, j \in \set{\pm 1}}\]
    and whose edges are
    \[ \what{E} = \sett{\set{v^j,u^i}}{\set{v,u} \in E, ij=f(\set{u,v})}.\]
    The lift map is \(\phi(v^j) = v\).
\end{definition}
It is elementary to prove this construction is indeed a lift, so we omit this proof. It is also easy to show that any \(2\)-lift is an induced lift for some signing \(f:V(G) \to \set{\pm 1}\). See \cite{Surowski1984} for a more general statement and proof.

In \cref{sec:construction} we generalize the notion of graph lifts to \emph{local lifts} of simplicial complexes. 

\subsubsection{Signing Functions and Lift Expansion}
Fix a graph \(G\) and a signing \(f:E(G) \to \set{\pm 1}\). In this subsection, we characterize the eigenvalues of an \(f\)-induced lift. For this, we need to define the \(f\)-signing of an adjacency operator. For a \(d\)-regular graph the \(f\)-signing of the adjacency operator is the matrix \(A^f(u,v) = f(u,v)\cdot A(u,v)\) for \(\set{u,v} \in E\) and \(A^f(u,v)=0\) if \(\set{u,v} \notin E\). 

This signing matrix is closely related to the random walk operator of the lift. In particular, the following is by now classical.

\begin{lemma} \label{lem:bl_union}
    Let \(G\) be a \(d\)-regular graph and let \(\what{G}\) be an \(f\)-induced \(2\)-lift. Then the eigenvalues of \(A_{\what{G}}\) are the union (with multiplicities) of the eigenvalues of \(A\) and those of \(A^f\). In particular, if \(\norm{A^f} \leq \lambda\) and \(G\) is a \(\lambda'\)-two sided (resp.\ one sided) spectral expander, then \(\what{G}\) is a \(\max\set{\lambda, \lambda'}\)-two sided (resp.\ one sided) spectral expander.
\end{lemma}

Using this lemma, \cite{BiluL2006} gave a criterion for the expansion of the lift graph.
\begin{lemma}[{\cite[Lemma 3.3]{BiluL2006}}]\label{lem:father-of-inverse-mixing-lemma}
Let \(G,f\) and \(A^f\) be as above and assume that \(G\) is a \(\lambda\)-two sided (resp.\ one sided) spectral expander with no self loops. Assume that for any pair of disjoint $S,T\subseteq V(G)$
\begin{equation*}
\Abs{\Iprod{\ovec_S,A^f\ovec_T}} \le \alpha\sqrt{|S||T|},
\end{equation*}
    then \(\what{G}^f\) is a \(\lambda'\)-two sided (resp.\ one-sided) spectral expander where \(\lambda' = \max \set{ \lambda, O\Paren{\alpha \Paren{1+\log\frac{1}{\alpha}}}}\).
\end{lemma}
We note that there is a nice formula for this inner product, which is \[\Iprod{\ovec_S,A^f\ovec_T}= \frac{1}{d}\sum_{(v,u): \set{v,u} \in E(G)}f(u,v)\indvec{S}(v) \indvec{T}(u).\]

\subsection{High Dimensional Expanders} \label{sbsc:hdx-prelims}

\begin{definition}[simplicial complex] A \emph{$k$-dimensional simplicial complex} is a finite hypergraph \(X\) that is downwards closed to containment. That is, if \(\tau \in X\) and \(\s \subseteq \tau\) then \(\s \in X\).
\end{definition}
We write $X=X(-1) \sqcup X(0) \sqcup X(1)\sqcup \cdots \sqcup X(k)$, where \(X(\ell) = \sett{\s \in X}{|\s| = \ell+1}\) (here \(X(-1)=\set{\emptyset}\) is mainly a formality) and the maximal size of a set \(\s \in X\) is \(k+1\). We call elements \(\s \in X(\ell)\) \emph{\(\ell\)-faces}, and in this case we say that \(X\) is \(k\)-dimensional. In this paper we will always assume the simplicial complex in question is \emph{pure}, that is, that every \(\s \in X(\ell)\) contained in some \(\tau \in X(k)\). 
In addition, we assume it has no self-loops or multifaces. Namely, every vertex appears at most once in each face and any face appears at most once in $X$.

\paragraph{Degrees and regularity}The degree of a face \(\s \in X(i)\) is \(d(\s) = \abs{\sett{\tau \in X(i+1)}{\tau \supseteq \s}}\). We say that a family of simplicial complexes \(\set{X_i}_{i=0}^\infty\) is \emph{\(j\)-bounded degree} if there is an integer \(M > 0\) so that for all \(X_i\) and all \(\s \in X_i(j)\), \(d(\s) \leq M\). If the family is \(0\)-bounded degree, we sometimes just say bounded degree (without the zero).

\begin{definition}[hyper-regularity]
Let $d_0 \ge d_{1} \ge \cdots \ge d_{k-1}$ be positive integers. A $k$-dimensional simplicial complex $X$ is \emph{$(d_0,d_1,\ldo,d_{k-1})$-regular} if for any $i\in\{0,\ldo,k-1\}$ and any $i$-face $\s$, \(d(\s)=d_i\). 
\end{definition}
We say that \(X\) is regular if there exists such a tuple so that \(X\) is $(d_0,d_1,\ldo,d_{k-1})$-regular. In this case we denote by \(d_i(X)=d_i\).

\paragraph{Skeletons and links}The \emph{$j$-skeleton} of a simplicial complex $X$ is the simplicial complex obtained by taking all the $i$-faces of $X$, for all $i\le j$. The \(1\)-skeleton of a complex is also called an \emph{underlying graph}. 

A link is a generalization of a graph neighborhood.
\begin{definition}[link]
    For a $k$-dimensional simplicial complex $X$ and a face $\s\in X$, the \emph{link of $\s$} is the $(k-1-|\s|)$-dimensional simplicial complex 
    \[X_\s = \sett{\tau \setminus \s}{ \tau \in X , \tau \supset \s}.\]
\end{definition}
For \(\ell \leq k-2\) and \(\s \in X(\ell)\) we denote by \(A_\s\) the random walk operator of the \(1\)-skeleton of \(X_\s\).
We often abuse of notation and for a face $\s = \{v_0,\ldo,v_\ell\}\in X(\ell)$ write $\s = v_0\ldo v_\ell$.

Natural analogues of expander graphs to higher dimensions are simplicial complexes where the neighborhoods of the faces are themselves expander graphs. See \pref{sec:intro} for more context on this important definition.

\begin{definition}[$\lambda$-high dimensional expander]
    Let $\lambda\in[0,1]$. A $k$-dimensional simplicial complex $X$ is a \emph{$\lambda$-two sided (resp.\ one sided) high dimensional expander} if for all \(i \leq k-2\) and all $\s \in X(i)$, the 1-skeleton of $X_\s$ is a $\lambda$-two sided (resp.\ one sided) spectral expander.
\end{definition}

\subsection{The Lov\'asz Local Lemma}
The Lov\'asz Local Lemma is a classical result in the probabilistic method.
\begin{lemma}[Lov\'asz Local Lemma \cite{ErdosLovasz75}]\label{lem:Lovasz-Local-Lemma}
    Let $\Ifrak{B}=\{B_1,\ldo,B_n\}$ be a finite set of events in some arbitrary probability space.  
    The \emph{dependency graph of $\Ifrak{B}$} is a digraph $G_\Ifrak{B}=(\Ifrak{B},E)$ so that any event $B_i\in\Ifrak{B}$ is mutually independent of all the events $\Ifrak{B}\setminus \Gamma(B_i)$, where $\Gamma(B_i)$ is the neighborhood of $B_i$ in $G_{\Ifrak{B}}$.

    If there exists a real function $\rho: \Ifrak{B}\to [0,1)$ so that
    \begin{equation}\label{eq:prelim-bound-in-lovasz-local-lemma}
        \prin{B_i} \le \rho(B_i)\prod_{B_j\sim  B_i}\Paren{1-\rho(B_j)}  
    \end{equation}
    for any $B_i\in\Ifrak{B}$, then
    with strictly positive probability,
    none of the events $B_i$ occur.

\end{lemma}
This lemma also has an algorithmic version, first given in the seminal work of \cite{MoserT2010}. We give below a slightly less general version than the one in \cite{MoserT2010}.
\begin{lemma}[\cite{MoserT2010}] \label{lem:Lovasz-Local-Lemma-alg}
    Let \(\Omega\) be a finite set and let \(\Ifrak{P} = (P_1,P_2,\dots,P_m)\) be a tuple of independent random variables supported on \(\Omega^m\). Let \(\Ifrak{B} = \set{B_1,B_2,\dots,B_n}\) be a finite set of events in the sigma algebra of \(\Ifrak{P}\). Let the dependency graph and the assignment \(\rho: \Ifrak{B} \to [0,1)\) be as in \pref{lem:Lovasz-Local-Lemma}. Then there exists a randomized algorithm that finds an assignment \(p \in \Omega^m\) such that \(p \notin \bigcup_{i=1}^n B_i\). If one can verify whether \(B_i\) holds in time \(t\), then the randomized algorithm runs in \(tn \sum_{i=1}^n \frac{\rho(B_i)}{1-\rho(B_i)}\) expected time.

\end{lemma}
The algorithm described in this lemma is simple. The algorithm starts with randomly sampling some \(p \in \Omega^m\). While there exists some \(B_i\) such that \(p \in B_i\), the algorithm takes an arbitrary such \(B_i\), and resamples all the coordinates \(P_j\) that \(B_i\) depends upon. Of course, if the algorithm halts, then \(p \notin \bigcup_{i=1}^n B_i\). The paper \cite{MoserT2010} shows that the expected number of times an event \(B_i\) is resampled is at most \(\frac{\rho(B_i)}{1-\rho(B_i)}\) which explains the runtime of this algorithm.

%% file: sections/03construction.tex
\section{Local lifts}\label{sec:construction}

This section presents our basic construction, the \emph{local lift} of a complex. We will define this construction formally and describe some of its properties.

\begin{construction}[Local Lift]\label{const:base-construction}
    Let \(X\) be a $k$-dimensional simplicial complex and let $f:X(k)\to \{\pm 1\}$. The \emph{$f$-local lift} of \(X\) denoted by $\what{X} = \what{X}^f$, is the following \(k\)-dimensional simplicial complex: 
    \begin{itemize}
        \item \(\what{X}(0) = X(0) \times \set{\pm 1}\) and we denote the vertices by \(\what{X}(0) = \sett{v^j}{v \in X(0), j \in \set{\pm 1}}\).
        \item For any $1 \leq \ell \leq k-1$,
        \begin{equation*}
            \what{X}(\ell) = \Iset{\{v_0^{j_0},v_1^{j_1},\ldots,v_{\ell}^{j_{\ell}}\}}{\{v_0,v_1\ldots,v_{\ell}\}\in X(\ell),\  j_0,\ldo,j_{\ell}\in\{\pm1\}}.
        \end{equation*}
        \item Finally, $\what{X}(k)$ is the set of all faces \(\s=\set{v_0^{j_0},v_1^{j_1},\dots,v_d^{j_k}}\) so that the face without the signs is \(\set{v_0,v_1,\dots,v_k} \in X(k)\), and the product of the \(j_i\)'s are equal to \(f(\s)\). Namely,
        \begin{equation*}
            \what{X}(k) = \Iset{\left\{v_0^{j_0},v_1^{j_1},\ldo,v_{k}^{j_{k}}\right\}}{\{v_0,v_1,\ldo,v_{k}\}\in X(k) , \;\; f(\{v_0,v_1,\ldo,v_{k}\})=\prod_{i=0}^k j_i}.
        \end{equation*}
    \end{itemize}
\end{construction}
One can already see that the \((k-1)\)-skeleton of \(\what{X}\) doesn't depend on \(f\) and is just some inflation of the original complex. The dependence on \(f\) is only in the top-level faces. Thus, in particular, \(\what{X}\) is not a lift of (the underlying graph) of \(X\), except when \(k=1\). However, the links of \((k-2)\)-faces in \(\what{X}\) are lifts of links in \(X\), which is why we named this complex a local lift. We will see this in the next subsection.

\subsection{Local Properties of Local Lifts} \label{sbsc:local-lift-properties}
For the rest of this subsection, we fix $X$ to be a $k$-dimensional pure simplicial complex, $f:X(k)\to\prodft$ to be a signing function, and  $\what{X}$ to be the $f$-local lift of $X$. We also need the following three pieces of notation:
\begin{enumerate}
    \item Let $\pi:\what{X}(0)\to X(0)$ be the projection map \(\pi(v^j) = v\), and we extend it to higher dimensional faces as well by \(\pi(\set{v_0^{j_0},v_1^{j_1},\dots,v_i^{j_i}}) = \set{v_0,v_1,\dots,v_i}\).
    \item Let $\Isign : \what{X} \to \prodft$ be 
$\Isign(\hat{\s}) \to \prod_{v^j\in \hat{\s}}j$.\label{item:sign-function}
    \item For any \(\hat{\s} \in \what{X}(k-2)\) with \(\s = \pi(\hat{\s})\) we denote by \(f_{\s}:X_\s(1) \to \set{\pm 1}\) the function \(f_{\s}(e) = f(\s \sqcup e)\) and by \(f_{\hat{\s}}:X_\sigma(1) \to \set{\pm 1}\) the function \(f_{\hat{\s}}(e)=\sign(\hat{\s})\cdot f_\sigma(e)\).\label{item:const-induced-signing-on-edged}
\end{enumerate} 

The first observation is that the degrees of \(\what{X}\) are twice the degrees of \(X\), except for \(d_{k-1}\), which stays the same.
\begin{observation}
    \label{obs:mult_2}
    If $X$ is $(d_0,\ldo,d_{k-1})$-regular then $\what{X}$ is $(2d_0,2d_1\ldots,2d_{k-2},d_{k-1})$-regular. \(\qed\)
\end{observation}
It is a direct calculation, so we have omitted its proof. We just comment that the reason that \(d_{k-1}\) remains the same is that for every \(\hat{\s} \in \what{X}(k-1)\) and \(v \in X_{\pi(\hat{\s})}(0)\) there is exactly one \(j \in \set{\pm 1}\) such that \(\hat{\s} \cup \set{v^j} \in \what{X}(k)\). Therefore, \(d_{k-1}(X) = |X_{\pi(\hat{\s})}(0)| = |\what{X}_{\hat{\s}}(0)| = d_{k-1}(\what{X})\).

The next lemma gives a complete description of the links of \(\what{X}\).
\begin{lemma}[on the structure of the links]\label{lem:const-structure-of-links}
Let $\hat{\s}\in \what{X}$ and denote by $\s=\pi(\hat{\s})$. 
\begin{enumerate}
    \item If $\dim(\hat{\s}) < k-2$, then the 1-skeleton of $\what{X}_{\hat{\s}}$ (the $\hat{\s}$-link of $\what{X}$), is isomorphic to the 1-skeleton of $X_{\s}$ tensored with the complete graph on two vertices with self loops\footnote{Note that this is also true for the link of \(\sigma = \emptyset\), i.e.\ \(\what{X}\) itself.}.

    \item If $\dim(\hat{\s}) = k-2$, then 
    $\what{X}_{\hat{\s}}$ 
    is isomorphic to a lift of $X_\s$ induced by \(f_{\hat{\s}}\).
     \end{enumerate}
    
\end{lemma}

\begin{proof}
    The first item directly follows from the definition of a tensor product.
    For the second item, suppose $\dim(\hat{\s})=k-2$. Both the vertices of \(\what{X}_{\hat{\s}}\) and of the \(f_{\hat{\s}}\)-induced lift of \(X_\s\) are \(X_\s(0) \times \set{\pm 1}\). 
    As for the edges, \(\set{u^i,v^j} \in \what{X}_{\hat{\s}}(1)\) if and only if \(\set{u,v} \in X_\s(1)\) and \(ij \cdot \sign(\hat{\s}) = f(\s \sqcup \set{u,v})\) (or equivalently \(ij=f_{\hat{\s}}(\set{u,v})\)). This is precisely the relation that defines edges in the \(f_{\hat{\s}}\)-induced lift of \(X_\s\).
\end{proof}

The following corollary that bounds the spectrum of the links is direct.
\begin{corollary}\label{cor:consts-on-the-spectra-of-the-links}
    Let $\hat{\s}\in \what{X}$ and denote $\s=\pi(\hat{\s})$.
    \begin{enumerate}
    \item If $\dim(\hat{\s}) < k-2$ then $\lambda(\what{X}_{\hat{\s}})=\lambda(X_{\s})$.
        
    \item If $\dim(\hat{\s})=k-2$
        then $\spec{\what{X}_{\hat{\s}}}= \spec{A_\s}\cup\spec{A^{f_{\hat{\s}}}_\s}$, where $A_\s$ is the normalized adjacency matrices of $X_\s$ and $A_\s^{f_{\hat{\s}}}$ is its signed normalized adjacency matrix with respect to $f_{\hat{\s}}$.
    \end{enumerate}
\end{corollary}
\begin{proof}
    The first item follows from the first item in \cref{lem:const-structure-of-links} that shows the link of \(\what{X}_{\what{\sigma}}\) is isomorphic to the link of \(X_\sigma\) tensored with a complete graph, and \cref{fact:eigenvalues-of-tensor-product} that bounds the expansion of such a graph. The second item follows from \pref{lem:bl_union} and the fact that the link is the \(f_{\hat{\s}}\)-induced lift of \(X_\s\) as we saw in \pref{lem:const-structure-of-links}.
\end{proof}

%% file: sections/04existence.tex
\section{Families of HDXs via Random Local Lifts}\label{sec:existence}

This section is dedicated to existential proofs of high dimensional expanders based on our local lifts from \pref{const:base-construction}. 
We start by stating the main theorem of this section which asserts that given an arbitrary HDX $X$, there exists a family of HDXs with parameters comparable to those of $X$ so that any member of the family is a local lift of the former. Formally, 

\begin{theorem}
    \label{thm:exist-family-of-HDXs}
    Let \(X_0\) be a \((d_0,d_1,\dots,d_{k-1})\)-regular \(\lambda\)-two sided (resp.\ one sided) HDX over $n$ vertices, for $\lambda\in[0,1]$. Then there exists a family of $\max\set{\lambda,O\Paren{\sqrt{\frac{k^2\log^3 d_{k-1}}{d_{k-1}}}}}$-two sided (resp.\ one sided) high dimensional expanders $\{X_i\}_{i=0}^\infty$ 
    so that  
    $X_i$ is a $(2^id_0,\ldo,2^id_{k-2},d_{k-1})$-regular complex over $2^i n$ vertices and \(X_{i+1}\) is a local lift of \(X_{i}\). 
\end{theorem}

The proof of \pref{thm:exist-family-of-HDXs} is based on proving the single-step version of it, given in \pref{thm:exist-HDX-single-iteration}, and applying it iteratively. 

\begin{theorem}\label{thm:exist-HDX-single-iteration}
    Let $\lambda\in[0,1]$. 
    For any $k$-dimensional, $(d_0,\ldots,d_{k-1})$-regular, $\lambda$-two sided (resp.\ one sided) high dimensional expander $X$ over $n$ vertices, there exists a signing $f:X(k)\to\prodft$ so that $\what{X}$ is a  $\max\left\{\lambda,O\Paren{\sqrt{\frac{k^2 \log^3 d_{k-1}}{d_{k-1}}}}\right\}$-two sided (resp.\ one sided) high dimensional expander with regularity $(2d_0,\ldo,2d_{k-2},d_{k-1})$ and $2n$ vertices.
\end{theorem}

We start by proving \pref{thm:exist-family-of-HDXs} given \pref{thm:exist-HDX-single-iteration}. The proof of \pref{thm:exist-HDX-single-iteration} is more involved and is provided in the remainder of this section.

\begin{proof}[Proof of \pref{thm:exist-family-of-HDXs} assuming \pref{thm:exist-HDX-single-iteration}]
Let $X_0$ as specified in \pref{thm:exist-family-of-HDXs} and denote $\lambda' \dotdef\max\left\{\lambda, O\Paren{ \sqrt{\frac{k^2\log^3 d_{k-1}}{d_{k-1}} }}\right\}$.
The proof is by induction on $i$. 
Clearly $X_0$ holds the requirements. 

For the induction step, let $X_i$ be a $(2^id_0,\ldots,2^id_{k-2},d_{k-1})$-regular $\lambda'$-two sided (resp.\ one sided) HDX with $2^i n$ vertices received in the $i$-th step of the process. 
By \pref{thm:exist-HDX-single-iteration}, there exists a singing function $f_i:X_i(k) \to \prodft$ so that the $f_i$-local lift of $X_i$ (denoted by $\what{X_i}$) is a $ (2^{i+1}d_0,\ldo,2^{i+1}d_{k-2},d_{k-1})$-regular $\lambda'$-two sided (resp.\ one sided) HDX over $2^{i+1} n$ vertices. Setting $X_{i+1}\dotdef \what{X_i}$ concludes the proof.
\end{proof}

\subsection{Proof Outline of \pref{thm:exist-HDX-single-iteration}}\label{sec:exist-proof-of-main-thm-(single-HDX-existence)}

The proof of \pref{thm:exist-HDX-single-iteration} is based on Lov\'asz Local Lemma \cite{ErdosLovasz75} and \pref{lem:father-of-inverse-mixing-lemma}, and closely follows the lines of the existential proof in \cite{BiluL2006}.

Recall that one approach for proving a given $k$-dimensional simplicial complex is $\lambda$-HDX, is considering all of its $\ell$-links for $\ell \le k-2$ and bound the spectrum of each of their $1$-skeletons by $\lambda$. 
By \cref{cor:consts-on-the-spectra-of-the-links}, 
the only links one should be concerned with are those obtained by $(k-2)$-faces, 
as the links of all other faces 
 inherent the expansion from the initial HDX. In addition, by the same corollary, it's enough to analyze the spectra of the signed random walk matrices of the $(k-2)$-links of $X$, with respect to the signing induced on them 
as defined in \cref{lem:const-structure-of-links}. 
Indeed, doing so is the most technical part of the proof
and follows by the next lemma combined with \cref{lem:father-of-inverse-mixing-lemma}:  

\begin{lemma}\label{lem:exist-main-lemma-bounding-radius-for-pmz1-vectors}
    For any $k$-dimensional pure $(d_0,\ldo,d_{k-1})$-regular simplicial complex $X$ over $n$ vertices, there exists a signing function $f:X(k)\to\prodft$ such that for any $(k-2)$-face $\hat{\s}\in \what{X}$ and any disjoint subsets of vertices $S,T\subseteq X_\s(0)$ for $\s=\pi(\hat{\s})$, 
    \begin{equation}\label{eq:exist-eq-of-main-thechnical-lemma}
        \abs{\iprod{ \indvec{S}, A^{f_{\hat{\s}}}_\s\indvec{T}}} \le 10\sqrt{\frac{k^2\log d_{k-1}}{d_{k-1}}|S||T|}
    \end{equation}
    where 
    $f_{\hat{\s}}$ is the signing on $X_\s$'s edges induced by $f$ as defined in 
    \cref{sbsc:local-lift-properties}.
\end{lemma}

    By the $(d_0,\ldo,d_{k-1})$-regularity of $X$, $X_\s$ is a $d_{k-1}$-regular graph over $d_{k-2}$-vertices.
    Furthermore, since any signing over the $k$-faces induces a signing function on the edges of any $(k-2)$-link, our goal is to find a single signing function \(f\) such that these lifts of all the links of the \((k-2)\)-dimensional faces expand.

The proof of \pref{lem:exist-main-lemma-bounding-radius-for-pmz1-vectors} is by the Lov\'asz Local Lemma \pref{lem:Lovasz-Local-Lemma}. We define the set of ``bad'' events $\Ifrak{B}=\Ibrace{B^{S,T}_{\s}}$. The event \(B_{\s}^{S,T}\) is that \eqref{eq:exist-eq-of-main-thechnical-lemma} doesn't hold for a fixed $\s\in X(k-2)$ and fixed disjoint sets $S,T\subseteq X_\s(0)$. In \cite{BiluL2006}, similar bad events were considered, but only the sets \(S,T\) needed to be specified. The main difference between our proof and theirs is that we need to take care of the dependencies between events corresponding to different $(k-2)$-faces \(\s,\s'\).

To apply the lemma and deduce \pref{lem:exist-main-lemma-bounding-radius-for-pmz1-vectors}, 
one needs to understand and analyze the dependency relation of the events in $\Ifrak{B}$.

\paragraph
{On the dependency of bad events in $\Ifrak{B}$}
Fix $\hat{\s}\in \what{X}(k-2)$ and disjoint $S,T\subseteq X_\s(0)$ for $\s=\pi(\hat{\s})$, and define $F(\s,S,T)\subseteq X(k)$ to be the set of all $k$-faces of $X$ so that $\s\subseteq \tau$ and $\tau\setminus \s$ is an edge in the graph induces by $S\sqcup T$ on $X_\s$. 
Recall that 
$\Isign : \what{X} \to \prodft$ is defined by 
$\Isign(\hat{\s}) = \prod_{v^j\in \hat{\s}}j$, 
and note that 
    \begin{align*}
        d_{k-1} \iprod{\indvec{S},A^{f_{\hat{\s}}}_\s\indvec{T}}
        &= \!\!\!\!\!\!
        \sum_{
        \substack{uv\in X_\s(1)\\
        \mst u\in S, v\in T}}
        \!
        f_{\hat{\s}}(uv)
                        = 
        \sign(\hat{\s}) 
        \!\!\!
        \sum_{
        \substack{uv\in X_\s(1)\\
        \mst u\in S, v\in T}}
        \!
        f(\s \sqcup uv)
        = \sign(\hat{\s}) \!\!\!
        \sum_{\tau\in F(\s,S,T)}
        \!\!\!\!
        f(\tau).
    \end{align*}

    Since the signs $f$ assigns to the $k$-faces are chosen independently, if the event \(B_{\s}^{S,T}\) is \emph{not} mutually independent of a subset \(\Ifrak{A} \subseteq \Ifrak{B}\), there must exists some event \(B_{\s'}^{S',T'} \in \Ifrak{A}\) for which \(F(\s,S,T) \cap F(\s',S',T')\) is not empty.

    Towards using the Lov\'asz Local Lemma, we need to bound the probability of the event \(B_{\s}^{S,T}\) as well as the number of neighbors it has in the dependency graph considered in the Local Lemma. The bound on the probability follows directly from the arguments in \cite{BiluL2006}, but bounding the number of neighbors each event is more involved. In contrast to the expander graph case considered in \cite{BiluL2006} (where the ``bad'' events only depend on \(S\) and \(T\)), in simplicial complexes events corresponding to different faces \(\s,\s'\) (and therefore to different links) can depend on one another as long as they have a common $k$-face in \(F(\s,S,T)\cap F(\s',S',T')\). 
    A naive count of the number of such $k$-faces won't suffice, and would lead the bound in \pref{eq:exist-eq-of-main-thechnical-lemma} to scale with $d_{k-2}$.
    Therefore, we need to carefully characterize when exactly  \(F(\s,S,T)\) and \(F(\s',S',T')\) intersect.
    
    The first case where dependency can occur is when \(\s=\s'\). In this case, we are in the same setting as in \cite{BiluL2006}, which observed that there must be an edge in the subgraph induced by $S\cup T$ as well as in the one induced by $S'\cup T'$ for a dependency to happen.
    
    In the second case, which is new to our proof, $\s\neq \s'$ meaning that each of the events considers a different graph. 
    In this case, we observe that this implies both that there is a \(k\)-face \(\tau\) containing both \( \s,\s'\) and that either there exist vertices $v\in \s \cap (S'\sqcup T')$ and $s\in S\sqcup T$ so that $\tau\setminus\s'=\{v,s\}$, or that $\tau\setminus\s'\subseteq \s$.
    As we show below, this characterization is useful to bound the number of possible events that are dependent on a certain \(B_{\s}^{S,T}\). The following claim gives this characterization formally. 

\begin{claim}\label{claim:exist-caracterizing-dependent-events}
    Let an event $B^{S,T}_\s\in\Ifrak{B}$ and some subset $\Ifrak{A}\subseteq\Ifrak{B}$. If  
    \begin{itemize}
        \item for any $B^{S',T'}_{\s'}\in \Ifrak{A}$ with $\s=\s'$
        there is no edge lying in both $E_{X_\s}(S,T)$ and $E_{X_{\s}}(S',T')$, and
        
        \item for any $B^{S',T'}_{\s'}\in \Ifrak{A}$ with $\s\neq\s'$ there is no $k$-face $\tau\in X$ containing both $\s$ and $\s'$, so that $\tau\setminus \s$ and $\tau\setminus \s'$ are edges in  $E_{X_\s}(S,T)$ and $E_{X_{\s'}}(S',T')$ respectively, 
    \end{itemize}
    then $B^{S,T}_{\s}$ is mutually independent of $\Ifrak{A}$.
    \end{claim}

All left to conclude the proof of \pref{lem:exist-main-lemma-bounding-radius-for-pmz1-vectors} 
is carefully counting the number of events that fulfill one of the claim conditions and provide a real function that bounds the probability of each event as in \pref{eq:prelim-bound-in-lovasz-local-lemma}. We leave the details for the formal proof, which is given in the next section.

\subsection{Proving \pref{lem:exist-main-lemma-bounding-radius-for-pmz1-vectors}}\label{sbsc:exist-proof-of-main-lemma}

This subsection is dedicated to the proof of \pref{lem:exist-main-lemma-bounding-radius-for-pmz1-vectors} together with the subclaims it requires. We start by setting notations and highlighting features that will be needed for the proof.   

\paragraph{Notations} 
We say that sets $S,T\subseteq X_\s(0)$ \emph{induce a connected subgraph} if the subgraph obtained by projecting $X_\s$ on $S\cup T$ is connected. In addition, we write $E_{X_\s}(S,T)$ to indicate the set of edges between $S$ and $T$ in $X_\s$.  
For $\hat{\s}\in\what{X}(k-2)$, we denote $\s=\pi(\hat{\s})$.
When the face $\hat{\s}$ being considered is clear from the context we abbreviate and write $f$ for $f_{\hat{\s}}$.

In addition, we rely on the following observation:

\begin{observation}\label{obs:exist-ind-values-for-edges}
    If a signing $f:X(k)\to\prodft$ independently assigns a uniform sign to each $k$-face, then for any $\hat{\s}\in \what{X}(k-2)$, $f_{\hat{\s}}$ independently assigns a uniform sign to each edge in $X_\s$.  
\end{observation}
\begin{proof}
    Let $\hat{\s}\in \what{X}(k-2)$ and let $e\neq e'\in X_\s(1)$. Then for any $j,j'\in\prodft$
    \begin{align*}
        \prin{f_{\hat{\s}}(e)= j \wedge f_{\hat{\s}}(e')=j'} 
        &= \prin{\sign(\hat{\s})\cdot f(\s\sqcup e)= j \wedge 
        \sign(\hat{\s})\cdot f(\s \sqcup e')=j'} \\
        &= \prin{\sign(\hat{\s})\cdot f(s\sqcup e)= j}\cdot\prin{\sign(\hat{\s})\cdot f(\s \sqcup e')=j'}
        = \frac14.
    \end{align*}
\end{proof}
In addition, as in \cite{BiluL2006}, we can restrict the proof to consider only a pair of sets inducing connected subgraphs and deduce the result to any pair of sets.
In addition, we can assume 
that $d$ is as large as $996$ as it is always the case that 
$\Iprod{\indvec{S},A_\s^{f}\indvec{T}} \le \sqrt{|S||T|}$
and $1\le 10\sqrt{\frac{k^2\log d_{k-1}}{d_{k-1}}}$ for $d_{k-1} \in [1,996]$.

We are now ready to prove \pref{lem:exist-main-lemma-bounding-radius-for-pmz1-vectors}. 

\begin{proof}[Proof of \pref{lem:exist-main-lemma-bounding-radius-for-pmz1-vectors}]

    We set $f$ to be a randomized signing of $X(k)$ by setting a uniform and independent sign from $\prodft$ to any $k$-face.  
    Fix some face $\hat{\s}\in \what{X}(k-2)$ with \(\pi(\hat{\s})=\s\), and disjoint sets $S,T\subseteq X_\s(0)$. Denote the ``bad'' event in which the claim does not hold for our fixed face and sets by $B_\s^{S,T}$. That is, 
    \begin{equation*}
        \prin{B_\s^{S,T}} = \prin{\abs{\iprod{\indvec{S},A^{f}_{\s }\indvec{T}}} 
        > 10\sqrt{\frac {k^2\log d_{k-1}} {d_{k-1}}  |S||T|}}.
    \end{equation*}

    Fix for a moment some edge $uv\in X_\s(1)$, and consider the $(u,v)$-th entry of $A^{f}_\s$. By \pref{def:function-induced-lift}, $A^{f}_{\s}(u,v)= \frac{1}{d_{k-1}} f_{\hat{\s}}(uv)$, which, per 
    \pref{obs:exist-ind-values-for-edges}, distributed uniformly in $\prodft$ and is independent of all other edges signs. In addition, since 
    \begin{align*}
        \iprod{\indvec{S}, A^{f}_\s \indvec{T} }
        = \frac{1}{d_{k-1}}
        \sum_{\substack{uv\in E_{X_\s}(S,T)}}  f_{\hat{\s}}(uv),
    \end{align*}
    $\iprod{ \indvec{S}, A^{f}_\s \indvec{T} }$ is a sum of independent uniform random variables over $\prodft$, implying that 
    \begin{equation*}
    \EE_{f}\Brac{\iprod{\indvec{S}, A^{f}_\s \indvec{T} }} 
    = \frac{1}{d_{k-1}} \sum_{
    \substack{
    uv\in E_{X_\s}(S,T)
    }
    }
    \EE_{f}\Brac{
    f_{\hat{\s}} (uv)}
    = 0.
    \end{equation*} 
    Hence, by Hoeffding's inequality,
    \begin{align}
        \prin{B^{S,T}_\s} 
        = \prin{\Abs{\iprod{\indvec{S},A^{f}_{\s}\indvec{T}}} > 10\sqrt{\frac{ k^2\log d_{k-1}}{d_{k-1}} |S||T|}}  
        \le 2\exp\Paren{-\frac{2 \cdot 100  \frac{k^2\ln d_{k-1}}{d_{k-1}}|S||T|}
        {\underset{
        \substack{
    uv\in E_{X_\s}(S,T)
    }
    }
    {\sum}
    \!\!
    \Paren{\frac1{d_{k-1}}-(-\frac1{d_{k-1}})}^2}}\label{eq:proof-random-lift-first-bound-bad-event}.
    \end{align}
    Assuming w.l.o.g. that $|S|\ge |T|$ we get  
    \begin{align*}
    \frac{200  \frac{k^2\ln d_{k-1}}{d_{k-1}}|S||T|}
        {\underset{
        \substack{
    uv\in E_{X_\s}(S,T)
    }
    }
    {\sum}
    \!\!
    \Paren{\frac1{d_{k-1}}-(-\frac1{d_{k-1}})}^2}
    &=
        \frac{200 k^2d_{k-1}(\ln d_{k-1})|S||T|}{4|E_{X_\s}(S,T)|}\\
        &\geq \frac{50k^2 d_{k-1}(\ln d_{k-1})|S||T|}{d_{k-1}|T|}\\
        &\ge 25k^2 \ln d_{k-1}\abs{S\sqcup T}.
    \end{align*}
    Hence, denoting $\Icard =|S\sqcup T|$,
    \begin{align}\label{eq:final-bound-for-pr(B)}
        \pref{eq:proof-random-lift-first-bound-bad-event} 
        \le 2\exp\Paren{-25ck^2\ln d_{k-1}}
        \le d_{k-1}^{-10\Icard k^2}.
    \end{align}
    We turn to analyze the dependency graph of the ``bad'' events:\footnote{Recall that the dependency graph of a set of events $\Ifrak{B}$ is a digraph with a vertex for each event $B\in\Ifrak{B}$ and any event $B$ is mutually independent of $\Ifrak{B}\setminus \Gamma(B)$.} Recall that  
    $
    \Ifrak{B}
    $ 
    is the set of all events $B_{\s}^{S,T}$ for $\s\in X(k-2)$ and disjoint subsets $S,T\subseteq X_\s(0)$.
    Using the characterization of correlated events in $\Ifrak{B}$ given in \pref{claim:exist-caracterizing-dependent-events}, we get the following bound on the neighborhood size of the events in the dependency graph:
    \begin{claim}\label{claim:exist-bounding-number-dependent-events}
        Let $B^{S,T}_\s \in\Ifrak{B}$ and denote $\Icard=|S\sqcup T|$. Then $B^{S,T}_\s$ has at most $3k^2\Icard d^{\Icard'-1}$ neighbors $B^{S',T'}_{\s'}$ with $|S'\sqcup T'|=\Icard'$.
    \end{claim}  

     Now, to apply Lovász Local Lemma, one needs to define a function $\rho:\Ifrak{B}\to[0,1)$ such that
    \begin{equation*}
        \prin{B^{S,T}_\s} \le \rho\paren{B^{S,T}_\s} \prod_{B^{S',T'}_{\s'}\sim B^{x,y}_{\s}}\Paren{1- \rho\paren{B^{S',T'}_{\s'}}}.
    \end{equation*}

    Set $\rho\paren{B^{S,T}_\s} = d_{k-1}^{-3\Icard k^2}$. Indeed
    \begin{align}
        \rho\paren{B^{S,T}_\s} \prod_{B^{S',T'}_{\s'}\in \sim B^{S,T}_{\s}}\Paren{1- \rho\paren{B^{S',T'}_{\s'}}}
        &= d_{k-1}^{-3 \Icard k^2} \prod_{B^{S',T'}_{\s'}\sim B^{S,T}_{\s}}\Paren{1- d_{k-1}^{-3|S'\cup T'|k^2}}\nt\\
        &= d_{k-1}^{-3 \Icard k^2} \prod_{\Icard'\in[n]}\Paren{1- d_{k-1}^{-3 \Icard' k^2}}^{2^{\Icard'}3\Icard k^2d_{k-1}^{\Icard'-1}}\label{eq:bounding-number-of-vectors}\\
        &\ge d_{k-1}^{-3 \Icard k^2}\exp\Paren{-6\Icard k^2 \sum_{\Icard'\in[n]}2^{\Icard'} d_{k-1}^{\Icard'-1}d_{k-1}^{-3 \Icard' k^2}}\label{eq:euler-identity}\\
        &\ge d_{k-1}^{-3 \Icard k^2} e^{-7ck^2}\nt\\
        &\ge d_{k-1}^{-10 \Icard k^2} \geq \prob{B_{\s}^{S,T}}\label{eq:taking-d-large}
    \end{align}
    where \pref{eq:bounding-number-of-vectors} is since for any $U\subseteq X_\s(0)$ of cardinality $c'$, there are at most $2^{c'}$ pairs of disjoint sets $S',T'$ with $S'\sqcup T'=U$, \pref{eq:euler-identity} is by $e^{-x}\le 1- \frac{x}{2}$ for any $x\in[0,1.59]$ and \pref{eq:taking-d-large} is by taking $d_{k-1} \ge 3$. Together with \pref{eq:final-bound-for-pr(B)}, this concludes the proof.
\end{proof}     

We are left to prove \pref{claim:exist-caracterizing-dependent-events} and \pref{claim:exist-bounding-number-dependent-events}.

 \begin{proof}
 [Proof of \pref{claim:exist-caracterizing-dependent-events}]

    Recall that 
    \begin{align}\label{eq:char-inprod-by-sum-over-k-faces}
         \iprod{\indvec{S},A^{f}_\s\indvec{T}} 
        = 
        \frac{\sign(\hat{\s})}{d_{k-1}}
        \sum_{\tau\in F(\s,S,T)} 
        f(\tau).
    \end{align}
    If \(B_{\s}^{S,T}\) is not mutually independent of \(\Ifrak{A}\) then there exists some $B^{S',T'}_{\s'}\in \Ifrak{A}$ so that \(F(\s,S,T) \cap F(\s',S',T') \ne \emptyset\). Hence it is enough to show that if the two items in \pref{claim:exist-caracterizing-dependent-events} hold, then \(F(\s,S,T) \cap F(\s',S',T') = \emptyset\) for all $B^{S',T'}_{\s'}\in \Ifrak{A}$.
    
    \paragraph{Case $\s=\s'$} 
    Note that $B^{S,T}_{\s}$ and $B^{S',T'}_{\s}$ are events of the same link $X_\s$. Hence, if $E_{X_\s}(x,y)\cap E_{X_\s}(S',T') = \emptyset$ then \(F(\s,S,T) \cap F(\s,S',T') = \emptyset\).
    
    \paragraph{Case $\s\neq \s'$}
    Clearly, by \pref{eq:char-inprod-by-sum-over-k-faces}, if there is no $k$-face containing both $\s,\s'$ then \(F(\s,S,T) \cap F(\s',S',T') = \emptyset\). 
    Otherwise, let $\tilde{\tau}\in X(k)$ with $\s,\s'\subseteq \tilde{\tau}$ and note that  $f(\tilde{\tau})$ appears in both $\iprod{ \indvec{S},A^{f}_\s\indvec{T}}$ and $\iprod{ \indvec{S'},A^{f}_\s\indvec{T'}}$ (or equivalently \(\tilde{\tau} \in F(\s,S,T)\cap F(\s',S',T')\)), if and only if  $\tilde{\tau}\setminus \s \subseteq E_{X_\s}(S,T)$ and  $\tilde{\tau}\setminus \s' \subseteq E_{X_{\s'}}(S',T')$.   
    \end{proof}

    We now turn to bound the number of events fulfilling each of the cases in \pref{claim:exist-caracterizing-dependent-events}.     

    \begin{proof}[Proof of \pref{claim:exist-bounding-number-dependent-events}]
        Fix $S',T'\subseteq X_\s'(0)$ with $|S'\sqcup T'|=\Icard'$.
        Denote the subgraphs of $X_\s$ induced by $S,T$ by $G$ and the subgraph induced by $S',T'$ on $X_{\s'}$ by $G'$.
        
        We start with bounding the number of neighboring events that consider the same link as $B^{S,T}_\s$ does. 
        Since both events are over the same graph, we are in the same case as in \cite{BiluL2006}, and therefore use the same argument:  
        Recall that we restricted ourselves without loss of generality to connected subgraphs, and we consider the case that $G$ and $G'$ share a common edge.
        Hence, one can think of $G'$ as a subgraph of $X_\s$ with $c'$ vertices, with one of its edges in $G$. Therefore, one can bound the number of such subgraphs by the number of subtrees of $X_\s$ with $\Icard'$ vertices rooted in $G$.
        By \cite{Knuth68,FriezeM99}, the number of such subtrees is upper bounded by $\Icard d^{\Icard'-1}$.

        We turn to bound the other case, where $\s\neq \s'$.   Fix some $k$-face given by the second case of \pref{claim:exist-caracterizing-dependent-events} and denote it by $\tilde{\tau}$.  
        Observe that $\s,S,T$ and $\s',S',T'$ are strongly related: 
        there exist $ s,t\in S\sqcup T$ and $s',t'\in S' \sqcup T'$ so that $\tilde{\tau}=\s\sqcup \{s,t\}=\s'\sqcup \{s',t'\}$, and $\{s,t\}$, $\{s',t'\}$ are edges in $E_{X_{\s}}(S,T)$, $E_{X_{\s}}(S,T)$ respectively. Since $\s\neq \s'$, one of the following must hold:
        \begin{itemize}
            \item Either there exists $v\in\s\cap (S'\sqcup T')$ so that $\s'=(\s\setminus\{v\})\sqcup \{s\}$ 
            and therefore $\tilde{\tau}\setminus \s' = \{v,t\}$, 
        \end{itemize}
        \begin{itemize}
            \item or, both $s',t'\in \s \cap (S'\sqcup T')$ and $\s'=(\s\setminus\{s',t'\})\sqcup \{s,t\}$, 
            implying that $\tilde{\tau}\setminus \s' \subseteq \s$.
        \end{itemize}

        Hence, we can bound the number of events related to the first case by the number of subtrees of $X_{\s'}$ rooted in $\s$, that have a vertex in $S\sqcup T$ and another $\Icard'-2$ vertices which is at most $(k-1)cd^{\Icard'-2}$, and the events of the second case by the number of subtrees having two vertices from $\s$ (and rooted in one of them), and additional $\Icard'-2$ vertices which is at most $(k-1)^2d^{\Icard'-2}$. 
        
        And we get an overall bound of 
        \begin{align*}
            \Icard d^{\Icard'-1} 
            + 
            (k-1)\Icard d^{\Icard'-2} + (k-1)^2 d^{\Icard'-2} 
            \le 3k^2 c d^{c'-1}.
        \end{align*}
    \end{proof}

\subsection[Concluding the Theorems]{Concluding \pref{thm:exist-HDX-single-iteration}}\label{sbsc:exists-proof-of-main-thms}

\begin{proof}[Proof of \pref{thm:exist-HDX-single-iteration}]
    Let $X$ be a $(d_0,\ldo,d_{k-1})$-regular, $\lambda(X)$-two sided (resp.\ one sided) HDX over $n$ vertices, fix $f$ to be the signing provided by \pref{lem:exist-main-lemma-bounding-radius-for-pmz1-vectors}, and set $\what{X}$ to be the $f$-local lift of $X$.

    By \pref{obs:mult_2}, $\what{X}$ is a $(2d_0,\ldo,2d_{k-2},d_{k-1})$-regular graph over $2n$ vertices. We need to prove that for any $\hat{\s}\in \what{X}$ of dimension $\le k-2$, the 1-skeleton of $\what{X}_{\hat{\s}}$ is a $\max\left\{\lambda(X),O\Paren{ \sqrt{\frac{k^2\log^3 d_{k-1}}{d_{k-1}}}}\right\}$-two sided (resp.\ one sided)  expander.

    By \cref{cor:consts-on-the-spectra-of-the-links}, the spectra of all links $\what{X}_{\hat{\s}}$ with $\hat{\s}$ of dimension $< k-2$ are bounded by $\lambda(X)$.
    In addition, by \pref{lem:exist-main-lemma-bounding-radius-for-pmz1-vectors}, for any $\hat{\s}\in \what{X}(k-2)$ and any disjoint sets $S,T\subseteq X_\s(0)$ for $\s=\pi(\hat{\s})$, we have that $|\iprod{\indvec{S},A_\s^{f}\indvec{T}}| \le O\Paren{\sqrt{\frac{k^2\log d_{k-1}}{d_{k-1}} |S||T|}}$ where $A_\s^{f}$ is the $f_{\hat{\s}}$-signed random walk matrix of $X_\s$. Together with \pref{lem:father-of-inverse-mixing-lemma} this implies 
    \begin{align*}
        \lambda(X_\s)
        =O\Paren{ \sqrt{\frac{k^2\log d_{k-1}}{d_{k-1}} }\Paren{1+ \log\Paren{\sqrt{\frac{d_{k-1}}{k^2\log d_{k-1} }}}}}
        &\le        
        O\Paren{ \sqrt{\frac{k^2\log d_{k-1}}{d_{k-1}}}\cdot \log\sqrt{d_{k-1}}}\\
        &=        
        O\Paren{ \sqrt{\frac{k^2\log^3 d_{k-1}}{d_{k-1}}}}
    \end{align*}
    hence, by \pref{lem:bl_union}, $\lambda(X_\s)=\max\left\{\lambda(X), O\Paren{ \sqrt{\frac{k^2\log^3 d_{k-1}}{d_{k-1}} }}\right\}$ and by \pref{cor:consts-on-the-spectra-of-the-links}, this is also the case for $\lambda(\what{X}_{\hat{\s}})$.
\end{proof}

%% file: sections/05randomizedalg.tex
\section[An Algorithmic Version of Theorem 4.2]{An Algorithmic Version of \pref{thm:exist-family-of-HDXs}} \label{sec:rand-alg}
In this subsection, we prove that there is a randomized algorithm that finds a family of local lifts as in \pref{thm:exist-family-of-HDXs}
when \(X\) is a high dimensional expander under mild assumptions on the degree sequence which we encapsulate in the following definition.
\begin{definition}[Nice complex]
    Let \(X\) be a \(k\)-dimensional simplicial complex. We say that \(X\) is \emph{nice} if \(X\) is regular, and 
    \begin{equation}\label{eq:satisfy-degrees}
        d_{k-2}^{1-4\log d_{k-1}} < \frac{2}{e(k+1)kd_{k-1}+1}.
    \end{equation}
\end{definition}

We prove the following.
\begin{theorem} \label{thm:rand-alg-lll-full}
    There exists a randomized algorithm \(\mathcal{A}\) that takes as input a \(k\)-dimensional complex \(X_0\) and an integer \(i \geq 1\), runs in expected time \(\poly((2^i|X_0(0)|)^k)\) at most, and outputs a \(k\)-dimensional complex \(X_i\) with \(2^i |X_0(0)|\) vertices. The algorithm has the following guarantee.
    
    If \(X_0\) is a nice \((d_0,\ldo,d_{k-1})\)-regular \(\lambda\)-two sided high dimensional expander, then \(X_i\) is a \((2^i d_0,\dots,2^i d_{k-2},d_{k-1})\)-regular \(\lambda'\)-two sided high dimensional expander where 
    \[\lambda' = O \left ( \max  \set{\lambda\Paren{1+\log \frac{1}{\lambda}}, \sqrt{\frac{\log^3 d_{k-1}}{d_{k-1}}}}\right )\footnote{We will not calculate the constants in the big \(O\) notation explicitly.}.\]
          
    Moreover, one can modify the algorithm so that it outputs a sequence \(X_1,X_2,\dots,X_i\) of complexes all satisfying the same guarantees (instead of just the last one), so that for every \(j=0,1,\dots,i-1\), \(X_{j+1}\) is a local lift of \(X_j\).
\end{theorem}

Loosely speaking, in order to prove \pref{thm:rand-alg-lll-full}, we need to prove that there is an algorithm \(\mathcal{A}\) that finds a single local lift for \(X\) in polynomial time (just as \pref{thm:exist-family-of-HDXs} was proved by the `one-step theorem' \pref{thm:exist-HDX-single-iteration}) with good enough spectral expansion. Then we just iteratively use \(\mathcal{A}\) on its own output, setting \(X_{j+1} = \mathcal{A}(X_j)\), until reaching \(j=i-1\). For this to work, we also need to address the issue that \(\lambda' \geq \lambda\) so naively the expansion deteriorates as we reiterate, but this technical discussion is deferred to below.

\begin{remark}[A non-uniform algorithm for any HDX] \label{rem:degrees-eventually-satisfy}
\pref{thm:rand-alg-lll-full} requires that \(X_0\) be a nice complex, i.e.\ that \eqref{eq:satisfy-degrees} holds. However, in any family \(\set{X_i}_{i=0}^\infty\) where \(X_{i+1}\) is a local lift of \(X_i\), the degree \(d_{k-2}(X_i)\) tends to infinity with \(i\) while the other side of both inequalities stays fixed. 
Thus, the inequalities will eventually hold for any family of consecutive local lifts. In fact, they should hold for any \(i \geq C\log (k+d_{k-1}(X_0))\) for some large enough constant \(C > 0\). Thus we can modify the algorithm to work even if \(X_0\) is not nice (albeit with the spectral expansion bound guaranteed in \pref{thm:exist-family-of-HDXs}, which is slightly worse than the expansion in \pref{thm:rand-alg-lll-full}). This is done by allowing the algorithm to do a brute-force search for the first few steps, to produce a nice \(X_i\), and then continuing as the original algorithm does. The first few steps will eventually stop because \pref{thm:exist-family-of-HDXs} promises the existence of such an \(X_j\). This process takes \(\poly(|X_i|^k) + \exp(O(|X_0|^k))\) time.
\end{remark}

Towards the proof of \pref{thm:rand-alg-lll-full}, we need the following definition and lemma from \cite{BiluL2006}.

\begin{definition}\label{def:sparse}
    A graph $G$ with adjacency operator $A$ is said to be $(\beta,t)$-sparse if for every $S,T\subseteq V(G)$ such that $|S\cup T|\le t$, 
    \[
    \Iprod{\indvec{S},A\indvec{T}} \le \beta\sqrt{|S||T|}.
    \]
    For a $k$-dimensional hyper-regular complex $X$, we say that it is $(\beta,t)$-sparse if for every $\s\in X(k-2)$, the graph $X_{\s}$ is $(\beta,t)$-sparse.
\end{definition}
\begin{remark} \label{rem:connected-sparsity}
    While the definition here regards any \(S,T\) with $|S\cup T|\le t$, it is in fact equivalent to regarding only \(S,T\) with $|S\cup T|\le t$ \emph{such that the graph induced on \(S \cup T\) is connected}. We also remark that if \(X\) is \((\beta,t)\)-sparse then it is also \((\beta',t)\)-sparse for any \(\beta' \geq \beta\).
\end{remark}
The reason we need this definition of sparseness is that in a random local lift, sparseness does not deteriorate with high probability. More formally, the following lemma was proven in \cite{BiluL2006}.
\begin{lemma}[{\cite[Lemma 3.4]{BiluL2006}}]\label{lem:expander-is-sparse}
    Let \(G = (V,E)\) be a \(d\)-regular graph with \(n\) vertices that is \((\beta,\log n)\)-sparse for \(\beta \geq 10\sqrt{\frac{\log d}{d}}\). Then with probability \(\geq 1- n^{-4\log d}\) over \(f:E \to \set{\pm 1}\):
    \begin{itemize}
        \item For every \(S,T \subseteq V\), $\Abs{\Iprod{\indvec{S},A^f \indvec{T}}} \le \beta \sqrt{|S| |T|}$ and,
        \item \(\what{G}^f\) is $(\beta,\log{n}+1)$-sparse,
    \end{itemize}
\end{lemma}
We comment that \cite[Lemma 3.4]{BiluL2006} does not explicitly calculate the probability \(1-n^{-4\log d}\); rather, they only say the events happen with high probability. This is the probability that is implicit in their proof. They also prove this theorem for \(\beta = 10\sqrt{\frac{\log d}{d}}\) but the same proof extends to \(\beta \geq 10\sqrt{\frac{\log d}{d}}\) with no additional changes.

This next claim easily follows from the definition of expansion and says that a spectral expander is sparse.
\begin{claim}\label{claim:expander-is-sparse}
    Let $G$ be a $d$-regular $\lambda$-two sided spectral expander over $n$ vertices such that \(\lambda > \frac{1}{\sqrt{d}}\) and \(d \geq 3\). Then $G$ is $(2\lambda,\log{n})$-sparse.
\end{claim}
\begin{proof}
    Fix \(S,T\) such that with $|S\cup T|\le \log n$. By the \(\lambda\)-expansion and the expander mixing lemma (see e.g.\ \cite{HooryLW2006}), \(\iprod{\indvec{S},A\indvec{T}} \leq \frac{|S| |T|}{n} + \lambda \sqrt{|S| |T|}\). We bound this term by \( \left (\lambda + \frac{\log n}{n}\right ) \sqrt{|S| |T|}\). As \(\frac{\log n}{n} \leq \frac{1}{\sqrt{n}} \leq \frac{1}{\sqrt{d}} \leq \lambda\) the claim follows.
\end{proof}

We are ready to state our one-step theorem.
\begin{theorem} \label{thm:rand-alg-lll-one-step}
    There exists a randomized algorithm \(\mathcal{A}\) with the following guarantees. Let \(X\) be a $k$-dimensional \(\bar{d}\)-regular \(\lambda\)-two sided (resp.\ one sided) high dimensional expander over $n$ vertices, where $\bar{d}=(d_0,\ldo,d_{k-1})$. Let \(\beta \geq 10\sqrt{\frac{\log d_{k-1}}{d_{k-1}}}\) and denote by \(\lambda' = \max \set{\lambda, O(\beta (1+\log\frac{1}{\beta}))}\). Assume that \(X\) is \((\beta, \log d_{k-2})\)-sparse, and suppose that \(d_{k-2}\) and \(d_{k-1}\) satisfy \(d_{k-2} > d_{k-1}^2\) and \eqref{eq:satisfy-degrees}.
    Then \(\mathcal{A}(X) = \what{X}\) is a local lift of \(X\) such that:
    \begin{enumerate}
        \item The complex \(\what{X}\) is a \(\lambda'\)-two sided (resp.\ one sided) high dimensional expander.
        \item The complex \(\what{X}\) is \((\beta,\log 2d_{k-2})\)-sparse.
    \end{enumerate}
    Upon input \(X\) satisfying the above, the algorithm runs in time \(\poly(|X(0)|^k)\).
\end{theorem}

Before proving \pref{thm:rand-alg-lll-one-step}, let us quickly prove \pref{thm:rand-alg-lll-full}.
\begin{proof}[Proof of \pref{thm:rand-alg-lll-full}, assuming \pref{thm:rand-alg-lll-one-step}]
    Let \(X_0\) be as in the assumptions of \pref{thm:rand-alg-lll-full}. Let \(\mathcal{A}\) be the randomized algorithm in \pref{thm:rand-alg-lll-one-step}, and let \(X_{i+1} = \mathcal{A}(X_i)\) (or \(X_{i+1} = \bot\) if one of \(\mathcal{A}(X_j)\) fails to terminate for \(j \leq i\)). It suffices to show the following two statements:
    \begin{enumerate}
        \item If \(X_{i} \ne \bot\), then \(X_{i}\) satisfies:
        \begin{enumerate}
           \item \(X_i\) is a \(\lambda'\)-two sided high dimensional expander.
        \item \(X_i\) is \((\beta,\log 2d_{k-2})\)-sparse for \(\beta = \max \set{2\lambda, 10\sqrt{\frac{\log d_{k-1}}{d_{k-1}}}}\).
        \item \(X_i\) is nice, namely \eqref{eq:satisfy-degrees} holds for \(X_i\).
        \end{enumerate}
        \item For every \(i \geq 1\), \(X_{i}\) is output in expected \(\poly(|X_{i-1}(0)|^k)\)-time. In particular, the probability that any \(\mathcal{A}(X_i) = \bot\) is \(0\).
    \end{enumerate}
    The first item is proven by induction, where the base case for \(X_0\) holds by assumption (the sparseness is also by \(\lambda\)-two sided spectral expansion and \pref{lem:expander-is-sparse}). Assume that \(X_{i-1}\) satisfies the three sub-items, and prove for \(X_i\). Indeed, \(X_i = \mathcal{A}(X_{i-1})\) satisfies the first two sub-items by \pref{thm:rand-alg-lll-one-step}, and satisfies the third item because \(d_{k-2}(X_{i}) = 2d_{k-2}(X_{i-1})\) and the rest of the parameters remain fixed. The inequality \eqref{eq:satisfy-degrees} remains true under increasing \(d_{k-2}\) so if they hold true for \(X_{i-1}\) they are also true for \(X_i\).

    As for the second item, note that assuming \(X_{i-1}\) satisfies the assumptions \pref{thm:rand-alg-lll-one-step}, the algorithm in \pref{thm:rand-alg-lll-one-step} runs in expected time \(\poly(|X_{i-1}(0)|^k)\).
\end{proof}

\begin{proof}[Proof of \pref{thm:rand-alg-lll-one-step}]  
    We intend to use \pref{lem:Lovasz-Local-Lemma-alg}. For this, we fix the following `bad' events \(\Ifrak{C} = \sett{C_\s}{\s \in X(k-2)}\) where \(C_\s \subseteq \set{f:X(k) \to \set{\pm 1}}\) is the event where : 
    \begin{enumerate}
        \item Either \(\what{X_\s}^{\pm f_\s}\) is \emph{not} a \(\lambda'\)-two sided spectral expander, or
        \item \(\what{X_\s}^{\pm f_\s}\) is \emph{not} \((\beta, \log 2d_{k-2})\)-sparse.
    \end{enumerate}
    
    By \pref{lem:expander-is-sparse} (and \pref{lem:father-of-inverse-mixing-lemma} that relates the first item in \pref{lem:expander-is-sparse} to spectral expansion), \(\Prob[f]{C_\s} \leq 2d_{k-2}^{-4\log d_{k-1}}\). Moreover, because every link of a \(\hat{\s} \in \what{X}(k-2)\) is a lift of $X_\s$ with respect to \(f_{\hat{\s}}\), then if none of the events \(C_\s\) occur, then \(\what{X}\) satisfies both items in \pref{thm:rand-alg-lll-one-step}. We will use \pref{lem:Lovasz-Local-Lemma-alg} to find such an assignment.
        
    We now construct a dependency graph for \(\Ifrak{C}\). 
    Let \(\s \in X(k-2)\) and \(U \subseteq X(k-2)\). The event \(C_\s\) only depends on \(f_\s\), so it only depends on \(k\)-faces \(\tau \supseteq \s\). Therefore, if \(C_\s\) and \(\sett{C_{\s'}}{\s' \in U}\) are \emph{not} mutually independent, then in particular there is a \(k\)-face \(\tau\in X\) and \(\s' \in U\) such that \(\tau \supseteq \s,\s'\). Hence, in our dependency graph we connect \(C_\s \sim C_{\s'}\) if there exists such a \(k\)-face containing both \(\s\) and \(\s'\). Let us upper bound the neighborhood size of an event \(C_\s\). The number of neighbors that $C_\s$ has is upper bounded by the number of $k$-faces containing $\s$ times \(\binom{k+1}{k-1}\) (the number of ways to choose \(\sigma' \subseteq \tau\)). Therefore, the number of neighbors is bounded by
    \[D \coloneqq \binom{k+1}{k-1} \cdot \abs{\sett{\tau \in X(k)}{\tau \supseteq \s}} = \binom{k+1}{k-1} \abs{X_\s(1)} = \frac{(k+1)k}{4}d_{k-2}d_{k-1}.\]
    
    By setting \(\rho: \Ifrak{C} \to [0,1)\) to be the constant function \(\rho(C_\sigma) = \frac{1}{D+1}\) we have that
    \[\prob{C_\s} \leq \rho(C_\s) \prod_{\sigma' \sim \sigma}(1-\rho(C_\s)),\]
    because \(\rho(C_\s) \prod_{\sigma' \sim \sigma}(1-\rho(C_\s)) \geq \frac{1}{D+1} \left (1-\frac{1}{D+1} \right )^{D} \geq \frac{1}{e(D+1)} \) and \(\prob{C_\s} 
    \leq 2d_{k-2}^{-4 \log d_{k-1}} 
    \le \frac{1}{e(D+1)}\) by \eqref{eq:satisfy-degrees}.

    Let us now verify that the algorithm in \pref{lem:Lovasz-Local-Lemma-alg} runs in polynomial time. We note that there the number of events in \(\Ifrak{C}\) is \(\poly(|X(0)|^k)\), and checking whether \(C_\s\) occurs could be done in \(\poly(|X(0)|)\)-time because it amounts to:
    \begin{enumerate}
        \item Find the spectrum of a signed adjacency operator of a \(d_{k-2}\)-sized graph.
        \item Going over all connected sets \(U \subseteq \what{X}_\s^{\pm f_\s}\)  of size \(\leq \log 2 d_{k-2}\) for every $\s\in X(k-2)$, finding \(S,T\) such that \(S \cup T = U\), and counting the number of edges between \(S\) and \(T\) to check if the pair \(S,T\) violates sparseness. There is a \(\poly(|X(0)|)\) such \(U,S,T\) at most. 
    \end{enumerate}
    Therefore, the randomized algorithm in \pref{lem:Lovasz-Local-Lemma-alg} will find a signing in \(\poly(|X(0)|^k) \cdot \sum_{\sigma \in X(k-2)} \frac{1}{D} = \poly(|X(0)|^k)\)-time. 
\end{proof}

\subsection[A 1-Bounded HDX Family with Arbitrarily Large Vertex Degree]{\(1\)-Bounded HDX with Arbitrarily Large Constant Vertex Degree}
\restatetheorem{thm:diameter-intro}

\begin{proof}[Proof sketch]
    Let \(\lambda > 0\). Let \(\tilde{\lambda}\) be such that \(\lambda \leq O \left (\tilde{\lambda}(1+\log \frac{1}{\tilde{\lambda}}) \right )\) where the constant in the big \(O\) comes from \pref{thm:rand-alg-lll-full}. The construction in \cite{FriedgutI2020} gives an infinite family of \(2\)-dimensional \(\tilde{\lambda}\)-two sided high dimensional expanders that are \((d_0,d_1)\)-regular, where \(d_1 \leq  \exp(\poly(\frac{1}{\tilde{\lambda}}))\).\footnote{See discussion in \cite[Section 4.3]{FriedgutI2020} on how to do so using \(2\)-transitive groups.} Denote such a family \(\set{Y_\ell}_{\ell=0}^\infty\). One can verify that because \(d_0 \geq d_1+1 \geq 3\) and \(k=2\), these complexes are indeed nice.

    Moreover, let us assume without loss of generality that  
    \[ \tilde{\lambda}\left (1+ \log \frac{1}{\tilde{\lambda}} \right ) \geq \frac{\log^3 d_{1}}{d_{1}},\]
    since the construction in \cite{FriedgutI2020} allows to increase the \(1\)-degree of the complex by a polynomial in \(d_1\) while maintaining the \(\tilde{\lambda}\)-expansion.
        
     Set \(i = \lceil \log (M/d_0(Y_\ell)) \rceil\) (this is a constant since \(d_0(Y_\ell)\) is independent of \(\ell\)). By applying the algorithm in \pref{thm:rand-alg-lll-full} with input \((Y_\ell,i)\), we can obtain a \(2\)-dimensional \(\lambda\)-two sided high dimensional expander \(X_\ell\) where \(d_1\) remains the same but \(d_0(X_\ell) = 2^i d_0(Y_\ell) \in [M,2M]\), thus proving the theorem.

     For the `in particular' part of the theorem statement, this just follows from the fact that the diameter of any \(d_1\)-regular graph with \(d_0\) vertices is at least \(D=\frac{\log d_0}{\log d_1}\).
\end{proof}

%% file: sections/06high-probability-construction.tex
\section{Derandomizing the Construction} \label{sec:deterministic}
In this section we provide a deterministic construction of \((k-1)\)-bounded families of high dimensional expanders, as referred to in \cref{thm:main-construction-deterministic}. For the rest of this section, we denote $\alpha_k(d)=10\sqrt{\frac{k^2\log d}{d}}$ (when $k$ is clear from context, we will write $\alpha(d)$).

We will prove the following theorem.
\begin{theorem}[Restatement of \pref{thm:main-construction-deterministic}]\label{cor:explicit_hdx}
    There exists a deterministic algorithm \(\mathcal{B}\) that takes as input a \(k\)-dimensional complex \(X_0\) and an integer \(i \geq 1\), runs in time \(\poly((2^i|X_0(0)|)^k)\), and outputs a \(k\)-dimensional complex \(X_i\) with \(2^i |X_0(0)|\) vertices. The algorithm has the following guarantee: If \(X_0\) is a \((d_0,\ldo,d_{k-1})\)-regular \(\lambda\)-two sided high dimensional expander, with \(\lambda>\alpha_k(d_{k-1})\), $d_{k-1}>2^{10k}$ and $|X_0(k-2)|\le d_{k-2}^{10k}$, then \(X_i\) is a \((2^i d_0,\dots,2^i d_{k-2},d_{k-1})\)-regular \(\lambda'\)-two sided high dimensional expander where 
    \[\lambda' = O \left (2^{5k}\lambda\Paren{1+\log \frac{1}{\lambda}}\right )\footnote{We will not calculate the constants in the big \(O\) notation explicitly.}.\] In particular, for every $n\in\NN$, choosing $i=\log n$ yields a complex with at least $n$-vertices.
\end{theorem}

This explicit construction generalizes the explicit construction for expanders given in \cite{BiluL2006}, which is based on the conditional probabilities method \cite[Chapter 16]{alon2016probabilistic}.

We first observe that under the assumption that the base complex is sparse (as in \cref{def:sparse}) and that \(|X(k-2)|\) is not too large, then a random local lift of \(X\) is also sparse and is a high dimensional expander with high probability. Then, we explain how we can find such a lift deterministically by greedily selecting the values \(f(\tau)\) one $k$-face at a time.

\begin{lemma}\label{lem:passover}
    Let $X$ be a $k$-dimensional, $(d_0,\ldo,d_{k-1})$-regular and $(\beta,\log{d_{k-2}})$-sparse simplicial complex
    so that \(\beta \geq \alpha(d_{k-1})\) and  $|X(k-2)|\le d_{k-2}^{\log {d_{k-1}}}$.
    
    Then, for $f:X(k)\to\{\pm 1\}$ drawn uniformly at random, with probability at least \(1- d_{k-2}^{-3\log {d_{k-1}}}\):
    
    \begin{enumerate}
        \item\label{itm:norm} For every $\s\in X(k-2)$ and every \(S,T \subseteq X_\s(0)\): $\Abs{\Iprod{\indvec{S},A^f_\s \indvec{T}}}\le \beta \sqrt{|S| |T|}$.
        \item\label{itm:sprs}
        The local lift $\what{X} = \what{X}^f$ is $(\beta,\log{d_{k-2}}+1)$-sparse. \(\qed\)
    \end{enumerate}
    
\end{lemma}
We comment that the condition $|X(k-2)|\le d_{k-2}^{\log {d_{k-1}}}$ may seem odd at first glance. However, similar to \pref{rem:degrees-eventually-satisfy}, this is eventually satisfied by every sequence \(\set{X_i}_{i=0}^\infty\) where \(X_{i+1}\) is a local lift of \(X_i\). Thus, we do not lose too much generality by assuming it.

The proof of \cref{lem:passover} follows by applying \pref{lem:expander-is-sparse} to every link and taking a union bound over the links. We omit the proof since it is a direct calculation.

The deterministic construction mentioned at the beginning of this section is composed of iterative applications of the local lift, where each application is according to the algorithm described in the following lemma.

\begin{lemma}\label{lem:single_step_deterministic}
    Let $X$ be a $k$-dimensional $(d_0,\ldo,d_{k-1})$-regular $(\beta,\log{d_{k-2}})$-sparse simplicial complex with $d_{k-1} > 2^{10k}$, \(\beta \geq \alpha(d_{k-1})\) and such that $|X(k-2)|\leq d_{k-2}^{10 k}$.
    
    Then, there is a deterministic $\poly{(|X(0)|^k)}$ time algorithm for finding a function $f:X(k)\to\{\pm 1\}$ such that:
    \begin{enumerate}
        \item \label{item:link_expn} For every $\s\in X(k-2)$, 
        $\|A^f_\s\|=O\Paren{2^{5k}\beta \left (1+\log\frac{1}{\beta} \right )}$.
        \item \label{item:link_sprs} 
        $\what{X}^f$ is $(\beta,\log{d_{k-2}}+1)$-sparse.
    \end{enumerate}
\end{lemma}

The proof of \cref{lem:single_step_deterministic} uses the method of conditional probabilities. The main idea is that, given the conditions on the input complex, we can define random variables denoted $Z^{(\s)}$, which serve as ``error'' indicators, where these errors occur with very small probability. By defining another set of random variables $Y^{(\s)}$ which correlate with the links' expansions, and amplifying the impact of each error, we are able to choose \(f(\tau)\) $k$-face by $k$-face, while tracking the expected value of the sum of those variables efficiently and making sure no error occurs.

We need the following claim, already implicit in \cite{BiluL2006}. 
\begin{claim}\label{clm:sparse_preserved}
    Let $X$ be a hyper-regular $k$-dimensional complex that is $(\beta,t)$-sparse. Then, for every $f:X(k)\to\{\pm 1\}$, $\what{X}^f$ is $(\beta,t)$-sparse.
\end{claim}

\begin{proof}[Proof of \cref{lem:single_step_deterministic}]

     Denote $m=d_{k-2}$, $d=d_{k-1}$ and $\beta'=\beta \left (1+\log\frac{1}{\beta} \right )$. Our assumption on $X$ then tells that $|X(k-2)|\leq m^{10 k}$. Fix $\s\in X(k-2)$, and let $r=2\lceil\log{m}\rceil$.
    
    For using the conditional probabilities method, we define a random variable $Q=Q(f)$ such that:\begin{itemize}
        \item $\E[Q(f)]$ can be calculated in $\poly{(m^k)}$ time when $f(\tau)$ is fixed on a subset of $\tau\in X(k)$ and uniformly random on the rest, and
        \item If $Q(f) \le \E[Q(f)]$ then $f$ is \emph{good}, that is, it satisfies \pref{item:link_expn} and \pref{item:link_sprs} in \pref{lem:single_step_deterministic}.
    \end{itemize}
    Once these conditions are met, we can iteratively assign values to $f$ on the $k$-faces, calculating the expectation after each assignment such that the conditional expectation remains smaller than the unconditioned expectation. As $|X(k)|\le|X(0)|^k$, the decision tree for the conditional probabilities method is of polynomial depth, and we will complete constructing a good $f$ in polynomial time. Henceforth, the rest of the proof is dedicated to defining $Q$, such that the above conditions hold. En route to this goal, we will define two sets of random variables, $Y^{(\s)}=Y^{(\s)}(f)$ and $Z^{(\s)}=Z^{(\s)}(f)$, which are intended to define our control over \pref{item:link_expn} and \pref{item:link_sprs}, respectively. The variable \(Q\) will be the sum over all \(Y^{(\s)},Z^{(\s)}\).
    
    For any closed walk $p$ of length $r$ in $X_\s$, we define a random variable $Y^{(\s)}_p$ that equals the product of the signs along $p$ according to $f$, divided by $d^r$. Note that the expectation of $Y^{(\s)}_p$ can be calculated in polynomial time, even when some values of $f$ have been fixed: if there is at least one unsigned edge with an odd number of appearances in $p$ then $\E\left[Y^{(\s)}_p\right]=0$, otherwise
    as $\E\left[Y^{(\s)}_p\right]$ is $\pm \frac{1}{d^r}$ according to $f$, even appearances cancel.
     We also define $Y^{(\s)}=\sum_{p} Y^{(\s)}_p$, and claim that $Y^{(\s)}=\Tr\Paren{(A^f_\s)^r}$ (where $\Tr(\cdot)$ is the trace function) since $\Tr\Paren{(A^f_\s)^r}$ is \(\frac{1}{d^{r}}\)-proportional to the sum over walks of length \(r\). As \(r\) is even, $\Tr\Paren{(A^f_\s)^r} \geq 0$ so \(Y^{(\s)}\) is a non-negative random variable. More over, the trace upper bounds $\norm{A^f_\s}^r$ hence upper bounding $Y^{(\s)}$ simultaneously for all $\s$ will enable us to satisfy \pref{item:link_expn}.

    By the first item in \cref{lem:passover} combined with \cref{lem:father-of-inverse-mixing-lemma}, with probability at least 
    $1- m^{-3\log {d}}$
    the lift is a \(\left ( C_0\beta' \right )\)-two sided high dimensional expander for some constant $C_0$. Since the trace of a matrix equals the sum of its eigenvalues, we get, in this case, that
    \[Y^{(\s)}=\Tr((A^f_\s)^r)\le m \left(C_0 \beta'\right)^r \leq \left(C_1 \beta'\right)^r,\]
    where $C_1$ is some constant.
   
    As for any $f$ it holds that $\|A^f_\s\|\le 1$, we get that $Y^{(\s)}(f)\le m$, and overall
    \[
    \E\left[Y^{(\s)}\right]
    \le \left(C_1 \beta'\right)^r+m^{-3\log {d}}\cdot m
    \le \left(C_2 \beta'\right)^r
    \]
    for some constant $C_2$.
    
    We turn towards defining our second random variable \(Z^{(\s)}\). 
    Let $\gamma$ be a constant defined later in the proof. 
    For every 
    $S,T \subseteq X_\s(0)$ with $|S\cup T|=\log{m}+1$ 
    and the graph induced by $S$ and $T$ on in $X_\s$ is connected, we define \[
    Z^{(\s)}_{S,T}=\begin{cases}
		\gamma, & \text{if } 
  \Abs{\Iprod{\indvec{S},A^f_\s \indvec{T}}} > \beta \sqrt{|S| |T|}, \\
		0, & \text{otherwise}.
	\end{cases}
    \]

    As in the proof of \cref{lem:exist-main-lemma-bounding-radius-for-pmz1-vectors} (see  \eqref{eq:final-bound-for-pr(B)}), the probability $\Pr[Z^{(\s)}_{S,T}>0]<d^{-10k^2\log{m}}$
    and therefore \[
    \E\left[Z^{(\s)}_{S,T}\right]\le d^{-10k^2\log{m}}\cdot \gamma.
    \]
    We also define $Z^{(\s)}=\sum_{S,T} Z^{(\s)}_{S,T}$.
    As in the proof of \cref{claim:exist-bounding-number-dependent-events}, there
    are at most $m(3d)^{\log{m}+1}<d^{3\log m}$ pairs $S,T$ with $|S\cup T|=\log{m}+1$ that induce a connected subgraph, since we can bound the number of such pairs using the number of subtrees of $X_\s$ with $\log{m}+1$ vertices rooted in $X_\s$ (times the number of ways to choose sets \(S,T\) that span the tree).
    Overall, we get that \[
    \E\left[Z^{(\s)}\right]\le d^{3\log m} \cdot\E\left[Z^{(\s)}_{S,T}\right]\le d^{-3k^2 r} \cdot\gamma.
    \]
    Note that like the variables $Y_p$, the expectation of $Z^{(\s)}_{S,T}$ can be calculated in polynomial time even for partial assignments for the values of $f$, as the number of $S,T$ involved in the summation is bounded by a polynomial in \(m^k\). Notice also that when $Z^{(\s)}(f)=0$, \pref{item:link_sprs} in \pref{lem:single_step_deterministic} holds true: $Z^{(\s)}_{S,T}$ implies that the sparsity inequality holds for all pairs \(S,T\) of size $\log{m}+1$ that induce a connected subgraph. For sets with support less or equal to \(\log m\), sparsity holds by \cref{clm:sparse_preserved}. Thus by \cref{rem:connected-sparsity} the lift is \((\beta,1+\log m)\)-sparse.
    
    We next define $Q^{(\s)}=Y^{(\s)}+Z^{(\s)}$, and set $\gamma=\left(C_2\beta' d^{3k^2}\right)^r$.
    Notice that \[
    \E[Q^{(\s)}]=\E[Y^{(\s)}]+\E[Z^{(\s)}] \le  2\left(C_2 \beta'\right)^r.
    \]
    
    We can now define $Q=\sum_{\s\in X(k-2)}Q^{(\s)}$ and get 
    \begin{equation}\label{eq:less_than_1}
    \E[Q] \leq |X(k-2)|\cdot 2\left(C_2 \beta'\right)^r
    \le 2m^{10k} \left(C_2 \beta'\right)^r
    \leq
    2\cdot2^{5kr} \left(C_2 \beta'\right)^r
    <\gamma. 
    \end{equation}  
    We claim that if \(Q\) is less or equal to its expectation, both \pref{item:link_expn} and \pref{item:link_sprs} hold. Indeed, $Z^{(\s)}=0$ for every $\s\in X(k-2)$, as the minimal non-zero value of $Z^{(\s)}$ is 
    $\gamma$.
    We further note that in this case, taking into account the non-negativity of the $Y^{(\s)}$-s,
    \[
    Y^{(\s)}\le Q\le \E[Q]
    \le 2\cdot2^{5kr} \left(C_2 \beta'\right)^r
    \]
    for every $\s\in X(k-2)$, and hence 
    \[\norm{A^f_\s}\le \left(Y^{(\s)}\right)^{\frac{1}{r}}
    \le 2 \cdot 2^{5k} \left(C_2 \beta'\right),
    \]
    therefore \pref{item:link_expn} is satisfied as well.
\end{proof}

We are now ready to prove our main result in this section.

\begin{proof}[Proof of \cref{cor:explicit_hdx}]
    Let \(\bar{d} = (d_0,d_1,\dots,d_{k-1}=d)\) and let $X_0$ be a $\bar{d}$-regular $\lambda$-two sided high dimensional expander for $\lambda>\alpha_k(d)$, such that  $|X_0(k-2)|\le d_{k-2}^{10k}$. By \pref{claim:expander-is-sparse}, it is also $(2\lambda,\log{d_{k-2}})$-sparse.
    
    Denote by $\mathcal{B'}$ the algorithm suggested by \cref{lem:single_step_deterministic}, and let \(X_1,X_2,\dots,X_i\) be such that $X_{j}=\mathcal{B'}(X_{j-1})$ for \(j\in [i]\).
    We set $X_i$ to be $\mathcal{B}$'s output.

    Let us show that, $X_i$ meets the guarantees of \cref{cor:explicit_hdx}. By \cref{obs:mult_2}, for every \(j \in [i]\), $X_j$ is \((2^{j} d_0,2^{j} d_1,\dots,2^j d_{k-2},d)\)-regular and \(|X_j(0)|=2^j |X_0(0)|\).
    
    In addition, one can verify by a direct calculation that, for any $j\in[i]$, \(|X_j(k-2)| = 2^{k-1} |X_{j-1}(k-2)|\), so if \(|X_{j-1}(k-2)| \leq d_{k-2}(X_{j-1})^{10k}\) then 
    \[|X_{j}(k-2)|=2^{k-1} |X_{j-1}(k-2)| \leq d_{k-2}(X_{j-1})^{10k}\cdot 2^{k-1} \leq d_{k-2}(X_j)^{10k}.\] 
    Thus, by induction and the fact that this inequality holds for \(X_0\), this holds for every \(j\).
    
    Finally, by \cref{lem:single_step_deterministic} one inductively obtains that for any $j$:
    \begin{enumerate}
        \item $X_j$ is an $O\Paren{2^{5k}\lambda \left (1+\log\frac{1}{\lambda} \right )}$-high dimensional expander.
        
        \item $X_j$ is $\Paren{2\lambda,\log{d_{k-2}(X_i)}}$-sparse.

        \item $X_j$ computed in time $\poly(|X_{j-1}(0)|^k)=\poly(2^{j-1}|X_{0}(0)|^k)$.
    \end{enumerate}
    as required.
\end{proof}